\documentclass[a4paper,USenglish,cleveref, autoref, thm-restate]{lipics-v2021}
\usepackage{xspace}
\usepackage{mathtools}
\usepackage{todonotes}
\usepackage{booktabs}
\usepackage{siunitx}
\usepackage{fnpct}
\usepackage{enumitem}

\newcommand{\flamecast}{FLaMECaST\xspace}

\todostyle{mg}{prepend,caption={MG},color=color3}
\todostyle{tb}{prepend,caption={TB},color=color0}
\todostyle{wy}{prepend,caption={WY},color=color1}

\DeclarePairedDelimiter{\set}{\{}{\}}
\DeclarePairedDelimiter{\norm}{\lVert}{\rVert}

\DeclareMathOperator{\cost}{cost}
\DeclareMathOperator{\opt}{opt}
\DeclareMathOperator{\len}{len}
\DeclareMathOperator{\dist}{dist}
\DeclareMathOperator{\pos}{pos}
\DeclareMathOperator{\vor}{Vor}
\DeclareMathOperator{\Tweber}{T_{{Weber}}}

\newcommand{\DP}{\operatorname{dp}}

\newlist{tabitemize}{itemize}{1}
\setlist[tabitemize]{nosep, left=0pt, label=--}

\hideLIPIcs  

\title{Flow-weighted Layered Metric Euclidean Capacitated Steiner Tree Problem} 

\author{Thomas Bläsius}{Karlsruhe Institute of Technology}{}{}{}
\author{Henrik Csöre}{Karlsruhe Institute of Technology}{}{}{}
\author{Max Göttlicher\footnote{\label{note:funding}Funded by German Research Foundation (DFG) as part of the Research Training
Group GRK 2153: Energy Status Data – Informatics Methods
for its Collection, Analysis and Exploitation.}}{Karlsruhe Institute of Technology}{}{}{}
\author{Elly Schmidt}{Karlsruhe Institute of Technology}{}{}{}
\author{Wendy Yi\footnotemark[1]}{Karlsruhe Institute of Technology}{}{}{}

\authorrunning{T. Bläsius, H. Csöre, M. Göttlicher, E. Schmidt, W. Yi} 

\Copyright{Thomas Bläsius, Henrik Csöre, Max Göttlicher, Elly Schmidt and Wendy Yi} 



\category{} 

\relatedversion{} 

\acknowledgements{}

\nolinenumbers 

\EventEditors{John Q. Open and Joan R. Access}
\EventNoEds{2}
\EventLongTitle{42nd Conference on Very Important Topics (CVIT 2016)}
\EventShortTitle{CVIT 2016}
\EventAcronym{CVIT}
\EventYear{2016}
\EventDate{December 24--27, 2016}
\EventLocation{Little Whinging, United Kingdom}
\EventLogo{}
\SeriesVolume{42}
\ArticleNo{23}

\begin{document}

\maketitle
\begin{abstract}
Motivated by hierarchical networks, we introduce the \textsc{Flow-weighted Layered Metric Euclidean Capacitated Steiner Tree} (\flamecast) problem, a variant of the Euclidean Steiner tree with layered structure and capacity constraints per layer.
The goal is to construct a cost-optimal Steiner forest connecting a set of sources to a set of sinks under load-dependent edge costs.
We prove that \flamecast is NP-hard to approximate, even in restricted cases where all sources lie on a circle.
However, assuming few additional constraints for such instances, we design a dynamic program that achieves a $\left(1 + \frac{1}{2^n}\right)$-approximation in polynomial time.
By generalizing the structural insights the dynamic program is based on, we extend the approach to certain settings, where all sources are positioned on a convex polygon.
\end{abstract}

\section{Introduction}

Many applications, like connecting households with electricity or water, connecting solar panels to the power grid, planning logistics systems or in general various facility location problems, boil down to finding a graph that satisfies certain needs while being cost efficient.
Often, a fundamental goal in these network design problems is to minimize the cost to connect all components.
While this goal can be modeled as the well-known Minimum Spanning Tree (MST) problem, it often fails to capture other requirements by the application.
A common generalization is the Steiner tree~\cite{hakimi1971}, which allows the use of Steiner vertices that can be -- but do not have to be -- included in the resulting graph.
Although this might seem like a small change, it makes a big difference from an algorithmic perspective: While the MST problem can be solved in polynomial time (see \cite{nesetril2001} for a translation of the original paper by Borůvka), the Steiner tree problem is NP-hard~\cite{Karp1972}.
This highlights that, while it is important to make simplifications and generalizations when formalizing and theoretically studying a problem, it is also crucial to study the different aspects relevant to applications.

In this regard, compared to the MST problem, the Steiner tree problem is a much better formalization of the challenges appearing in many practical problems.
However, it still fails to capture crucial aspects that appear in various applications.
To illustrate this in an example, assume we want to compute a cable layout for a solar farm, connecting the solar panels (sources) to transformers (sinks) that feed the produced electricity into the power grid.
This is usually done in a hierarchical fashion, where the Steiner vertices between the sources and the sinks represent different intermediate components that are used for monitoring, as safety precautions, or for bundling lower-level components (e.g., Y-connectors, (re)combiner boxes, inverters).
While it makes sense to abstract from the specific components, it would be an oversimplification to ignore the fact that these components have capacities depending on their level in the hierarchy.
Moreover, different edges in the resulting cable layout have to deal with different amounts of current.
As higher-capacity cables are more expensive, the cost of an edge not only depends on its length but also on the number of sources it connects to a sink.
This changes the cost of an edge from a local property that is determined by just the edge alone to a global property that depends on the topology of the network.
From an algorithmic standpoint, this seems to be a crucial difference.

We note that the above mentioned shortcomings are not exclusive to computing cable layouts of solar farms.
Other prominent examples of inherently hierarchical networks are health care systems~\cite{schultz1970} and production-distribution systems~\cite{kaufman1977}; also see the survey by \c{S}ahin and Süral on hierarchical facility location~\cite{sahin2007}.

Motivated by the hierarchical cabling layout of solar farms, Gritzbach, Stampa and Wolf~\cite{gritzbach2022} considered multiple layers with capacity-constrained components and defined a cable cost function that depends on the amount of current that runs through a cable.
However, the main difference to \flamecast is that candidate positions for the components are specified in the input.
In fact, we are not aware of work on such variants where the Steiner vertices can be placed arbitrarily in the plane.

 \begin{table}
\centering
\caption{Comparison of the related problems Discrete Facility Location (DFL), Continuous Facility Location (CFL), Multi-Sink-Multi-Source Steiner Tree (MSMSST), Hop-Constrained Steiner Tree (HCST), Gilbert Arborescence Problem (GAP) and Solar Farm Cabling Layout Problem (SoFaCLaP).
For DFL and CFL, we list some variants that have been researched.}
\label{tab:related_work}
\renewcommand{\arraystretch}{1.2} 
\begin{tabular}{
    m{2cm}  
    m{1.5cm} 
    m{1.4cm} 
    m{3cm}  
    m{4cm}  
}
\toprule
Problem & Layers & Positions & Capacities & Cost Function \\
\midrule
DFL~\cite{drezner2004} & \begin{tabitemize}
  \item none
  \item $k$~\cite{sahin2007}
\end{tabitemize} & discrete &
\begin{tabitemize}
  \item uncapacitated
  \item per facility~\cite{korupolu2000}
\end{tabitemize} &
\begin{tabitemize}
  \item per facility, per connection
  \item concave per facility~\cite{Soland1974}
\end{tabitemize} \\[10pt]
CFL~\cite{love1988} &
\begin{tabitemize}
  \item none
  \item $k$~\cite{narula1984}
\end{tabitemize} &
$\mathbb{R}^2$ &
\begin{tabitemize}
  \item uncapacitated
  \item for all facilities~\cite{carlo2012}
\end{tabitemize} &
per facility, edge lengths \\[10pt]
MSMSST~\cite{Westscott2023} & unlimited
& $\mathbb{R}^2$ &
--- &
edge lengths \\[10pt]
HCST~\cite{voss1999} & $k$
& discrete &
--- &
edge weights \\[10pt]
GAP~\cite{volz2013} & unlimited
& $\mathbb{R}^2$ &
--- &
$\sum \len(e) \cdot (d + h \cdot\mathrm{flow}(e)^\alpha$) \\[10pt]
SoFaCLaP~\cite{gritzbach2022} & $k$
& discrete &
per layer &
$\sum \len(e) \cdot \mathrm{flow}(e)$ \\[10pt]
\flamecast & $k$
& $\mathbb{R}^2$ &
per layer &
$\sum \len(e) \cdot \mathrm{flow}(e)^\alpha$ \\
\bottomrule
\end{tabular}
\end{table}

An overview of other problems that incorporate various aspects of \flamecast is presented in \cref{tab:related_work}.
In the most unconstrained variant of \flamecast with arbitrarily many uncapacitated intermediate layers and $\alpha = 0$, it is closely related to the directed Steiner tree problem~\cite{zosin2002}, where the goal is to connect given sources to a specified sink with minimal cost.
A variant of the Steiner tree problem that allows for multiple sinks is the Multi-Source-Multi-Sink Steiner tree problem~\cite{Westscott2023}.
Most Steiner tree problems in the Euclidean plane use the sum of edge lengths as cost; see the survey by Ljubić~\cite{ljubic2021}.
However, a Steiner tree variant called Gilbert Arborescence Problem uses a cost function that additionally depends on the weighted flow of an edge~\cite{volz2013}, similar to \flamecast.
A crucial difference of all Euclidean Steiner tree variants to \flamecast is that they do not allow a layered structure and capacitated Steiner vertices.

A problem that extends to hierarchical structures more easily is the facility location problem, where the location for facilities is chosen such that all demand points are satisfied and some cost function is minimized.
Variants of the facility location problem fall into two categories: in the discrete version, the locations are chosen from a given candidate set, while in the continuous version, the facilities can be arbitrarily placed in the Euclidean plane.
Both the discrete and the continuous variants are shown to be NP-hard~\cite{megiddo1984}.
As the core problem of various applications in a broad variety of subject areas such as energy networks~\cite{gokbayrak2017} and health care systems~\cite{rahman2000}, it has been widely researched, and a multitude of generalizations with different parameters have been proposed; see~\cite{drezner2004} and~\cite{love1988} for a survey.
Generalizations that are related to \flamecast are capacitated facilities, different cost functions, and multiple layers; see \cref{tab:related_work}.
To the best of our knowledge, edge costs depending on the amount of flow have not been considered much, in particular not from a theoretical perspective.
Moreover, there is comparably less work on combinations of different generalizations for the continuous version than on the discrete version.

With this paper, we introduce and initiate the study on the problem \flamecast\footnote{Short for \textsc{Flow-weighted Layered Minimal Euclidean Capacitated Steiner Tree}.}, a variant of the Steiner tree problem that formalizes the aspects of layers with capacities and edge cost depending on the number of connected sources.
We briefly introduce \flamecast here; see Section~\ref{sec:preliminaries} for a formal definition.
The input of \flamecast includes a set of sources and a set of sinks with fixed positions in the Euclidean plane.
The goal is to find a rooted forest that connects each sink (leaf) via Steiner vertices that can be freely positioned to one of the sinks (root).
The vertices in each layer (defined by the distance to the root) have to adhere to capacities, i.e., a vertex in layer $i$ can have at most $c_i$ leaves in its subtree, where $c_i$ is the capacity for layer $i$ given in the input.
The objective is to minimize the total edge cost, where the cost of an edge $e$ depends on its Euclidean length $\len(e)$ and its load $\ell(e)$, where $\ell(e)$ is the number of sources $e$ separates from a sink.
Specifically, we define the cost of $e$ to be $\len(e) \cdot \ell(e)^\alpha$, where $\alpha \in [0, 1]$ is a parameter.
Note that for $\alpha = 0$, this includes the case where the cost only depends on the length.
While $\alpha = 0$ gives a big incentive to bundle sources together with Steiner vertices to decrease the total edge length, the other extreme of $\alpha = 1$ has no incentive for bundling at all: Two sources sharing an edge higher up in the hierarchy makes the edge as expensive as if each source uses its own copy of the edge.
The intermediate range of $0 < \alpha < 1$ captures the subadditive behavior typical for many applications, where the cost of an edge increases with increasing load, but the increase becomes smaller for higher load.

\subparagraph*{Results}
We provide NP-hardness and approximation-hardness results for various settings and complement these with polynomial approximation algorithms in some more restricted settings.
See \cref{tab:results} for an overview.
In \flamecast, the goal is to find a forest (a so-called \emph{topology}) that satisfies all constraints such that an optimal \emph{embedding} of the topology is cost-optimal over all topologies.
Although there is no exact embedding algorithm even if there is only one Steiner vertex to be placed, there are efficient algorithms that find good approximations for this specific case.

\begin{table}
\centering
\renewcommand{\arraystretch}{1.05}
\setlength{\tabcolsep}{12pt}
\caption{Overview over the results for various settings.}
\begin{tabular}{ll}
\toprule
Setting & Result \\
\midrule

\textbf{General Positions}             & \\
\quad $\alpha = 1$           & algorithm in $O(n\cdot|T|^3)$ time\\
\quad $\lambda = 0$            & algorithm in $O(n\cdot|T|^3)$ time \\
\quad $\lambda \geq 1, \alpha \in [0, 1)$ & NP-hard to $\left(1 + \frac{1 + 6^\alpha - 7^\alpha}{14^\alpha n}\right)$-approximate \\[3pt]

\textbf{Circular Positions}            & \\
\quad \emph{Group-equally-spaced}   &  \\
\qquad $\lambda \geq 1$, $\alpha \geq \log_n\left(n - 1 + \cos\left(\frac{\pi}{n}\right)\right)$ & NP-hard \\[3pt]
\qquad $\lambda \geq 1$, $\alpha$ constant & NP-hard to $\left(1 + \frac{1}{n^2}\right)$-approximate \\
\quad \emph{Source-equally-spaced}  & \\
\qquad $\lambda = 1$, one sink           & $(1 + \frac{1}{n^2})$-approximation in $O(n^2\log^3(n))$ time \\[3pt]
\textbf{Convex Positions}              & \\
\quad $\lambda = 1$, $\alpha = 0$, one sink, no capacities & $(1 + \frac{1}{n^2})$-approximation in $O(n^8\log^3(n))$ time \\
\bottomrule
\end{tabular}
\label{tab:results}
\end{table}

For settings where the topology immediately fixes the embedding, we propose a polynomial algorithm that determines the optimal topology.
This is the case if there are no intermediate layers, i.e., no Steiner vertices are placed, but also if $\alpha = 1$.
However, we show that it is not only the embedding part that makes \flamecast hard.
In particular, we prove that for at least one intermediate layer and $\alpha < 1$, finding an optimal topology is NP-hard to approximate within a factor of $1 + \frac{1+ 6^\alpha-7^\alpha}{14^\alpha n}$, even assuming that we can compute the cost of an optimal embedding for a given topology exactly.

Moreover, we analyze settings where the given positions of the sinks and sources are restricted to circular or convex positions.
In circular instances, the sources are distributed on a circle, while the sinks are placed in the center of the circle.
In this case, we prove NP-hardness for instances with at least one intermediate layer if $\alpha$ is constant or sufficiently large.
For constant $\alpha$, we additionally show that it is NP-hard to approximate it within a factor of $1 + \frac{1}{n^2}$, again even if we assume optimal embeddings can be computed exactly.

However, if we additionally require all sources to be evenly distributed along the circle and there is only one sink, then these instances can be $\left(1 + \frac{1}{n^2}\right)$-approximated in polynomial time.
The algorithm we propose is a dynamic program that builds on useful structural properties of optimal topologies for such instances.
The approximation factor of our algorithm follows directly from that of the embedding problem of a single Steiner vertex, and our algorithm can reach exponential precision in polynomial time.
In fact, if we assume that optimal embeddings can be computed exactly in polynomial time, then our algorithm is exact.

Furthermore, we generalize the dynamic program to convex instances, where the sources lie on a convex polygon, with one sink and one intermediate layer.
The generalization only works to some extent, as we show that not all relevant properties of the circular case directly translate to general convex instances.
However, we prove that the properties still hold with some caveats in the uncapacitated case if $\alpha = 0$.
As in the circular case, the approximation factor for the embedding problem of a single Steiner vertex translates directly to our algorithm.

\section{Preliminaries}
\label{sec:preliminaries}
Let $G = (V, E)$ be a directed graph.
We denote a directed edge from $u \in V$ to $v \in V$ by $uv$.
A \emph{directed path} from a vertex $u \in V$ to a vertex $v \in V$ is a sequence of vertices $(v_0, \ldots, v_k)$ for some $k \in \mathbb{N}$ where $v_{i-1}v_i \in E$ for every $i \in [k]$ and $v_0 = u$ and $v_k = v$.

A directed graph is a \emph{rooted tree} if there is a vertex $r$ such that for every vertex $v$, there is a unique path from $v$ to $r$.
There is at most one such vertex in any directed graph, and we refer to it as the \emph{root}.
We call the length of the unique path from a vertex $v$ to the root $r$ the \emph{layer} of $v$.
By definition, every edge goes from a vertex in layer $i$ to a vertex in layer $i - 1$ for some $i \in \mathbb{N}_0$.

For an edge $uv \in E$, we say that $u$ is a \emph{child} of $v$ and $v$ is the parent of $u$.
Except for the root, every vertex has a unique parent.
A vertex that has no children is a \emph{leaf}.
The neighborhood of a vertex $v$ consists of the parent of $v$ (if it exists) and its children.
We denote the set of neighbors by $N(v)$.
The \emph{subtree} of a vertex $v \in V$ consists of all vertices that have a directed path to $v$.
The \emph{load} $\ell(v)$ of a vertex $v \in V$ is the number of leaves in its subtree.
Similarly, the load $\ell(e)$ of an edge $uv \in E$ is defined as the load of the child $u$.

A directed graph is a \emph{rooted forest} if it is a pairwise disjoint union of directed trees.
The \emph{height} of a rooted forest is the maximum layer over all vertices.
For given sets $S$ and $T$, we say that a rooted forest $F$ is an \emph{$S$-$T$-forest} if the leaves of $F$ are exactly $S$ and the roots of $F$ are a subset of $T$.

Given a tuple with \emph{capacities} $C = (c_0, \ldots, c_\lambda) \in \mathbb{N}^{\lambda + 1}$, we say that a forest with height at most $\lambda$ is \emph{valid} with respect to $C$ if $\ell(v) \leq c_i$ for every vertex $v$ in layer $i \in [\lambda]$.

We denote the \emph{distance} between two points $p, q \in \mathbb{R}^2$ by $\dist(p, q) = \norm{p - q}$.
Given a set of points $P \subseteq \mathbb{R}^2$, the \emph{Voronoi cell} $\vor(p)$ of a point $p \in P$ consists of all points for which there is no closer point in $P$ than $p$.

An \emph{embedding} of a directed forest $F$ is a function $\pos \colon V(F) \to \mathbb{R}^2$ that maps each vertex of $F$ to a point in the plane.
To simplify the notation, the distance $\dist(u, v)$ between two vertices $u$ and $v$ is the distance of $u$ and $v$ in the given embedding, i.e., $\dist(u, v) = \dist(\pos(u), \pos(v))$.
For some $\alpha \in [0, 1]$ and some forest $F$ with an embedding, the \emph{cost} of an edge $uv$ is $\cost(uv) = \dist(u, v) \cdot \ell(uv)^\alpha$.
The cost of the embedding is the sum of the cost over all edges.

\subparagraph{Problem Definition}
In the \textsc{Flow-weighted Layered Minimal Euclidean Capacitated Steiner Tree} (\flamecast) problem, an instance $I = (S, T, C, \alpha)$ consists of sources $S$ and sinks $T$, both with fixed positions in the plane, capacities $C = (c_0, \ldots, c_{\lambda + 1})$ with $c_0 \geq c_1 \geq \ldots \geq c_{\lambda + 1} \geq 1$\footnote{Note that requiring the capacities to be decreasing is not a real restriction.}, and a cost parameter $\alpha$.
The goal is to find an $S$-$T$-forest $F$ with height at most $\lambda + 1$ that is valid with respect to $C$ and an embedding of $F$ such that the cost of the embedding is minimum.

In such a forest, each source is connected to a sink via at most $\lambda$ inner vertices.
These inner vertices are called \emph{Steiner vertices}, and we say that the forest has $\lambda$ \emph{intermediate layers}.
For a given instance, an $S$-$T$-forest $F$ with height at most $\lambda + 1$ is called a \emph{topology}.
Together with an embedding, it is called a \emph{layout}.
A layout is \emph{valid} with respect to given capacities if the underlying topology is \emph{valid}.
In other words, the goal of \flamecast is to find a valid layout with minimum embedding cost.

We say that an instance is \emph{feasible} if there is a valid layout.
Observe that this is the case if and only if the total capacity of the sinks (layer $0$) is sufficiently high, i.e., $|T| \cdot c_0 \geq |S|$, since we can always directly connect each source with a sink.
\subparagraph{Optimal Position of a Steiner Vertex}
Consider a single Steiner vertex in a cost-optimal embedding of a given topology if the positions of its neighbors are already fixed.
To state properties on the possible position of such a Steiner vertex, we adapt results from Kuhn and Kuenne~\cite{kuhn1962} who analyzed a more general setting:
Given a set of vertices $P = \set{p_1, \ldots, p_n}$ with $n \geq 3$ and weights $w_1, \ldots, w_n > 0$ for each point, determine a point $p^*$ that minimizes the weighted sum $\sum_{i=1}^{n} w_i \dist(p^*, p_i)$.
Such a point is known to be unique, and it is called a \emph{Weber point}.
In our case, where the positions of the neighbors of a Steiner vertex $u$ are already fixed, each point $p_i$ represents the position of a neighbor $v_i$ and the weight $w_i$ is $\ell(uv_i)^\alpha$, the weighted load of the corresponding edge.
The Weber point is the optimal position of the Steiner vertex.
\begin{lemma}{\cite[Theorem 2]{kuhn1962}}
  \label{lem:properties:convex steiner}
  In an optimal embedding of a given topology for a \flamecast instance, each Steiner vertex lies in the convex hull of its neighbors.
\end{lemma}
The Weber point is also called the geometric median of a point set.
If more than half of the weights are concentrated on one single point $p$, then the Weber point is located at $p$.
\begin{lemma}{\cite[Theorem 3]{kuhn1962}}
  \label{lem:properties:majority-weber-point}
  Let $L$ be an optimal embedding of a topology for a given \flamecast instance.
  If there is a point $p$ with
  \begin{align*}
    \sum_{\substack{w \in N(v) \\ \pos(w) = p}} \ell(vw)^\alpha \geq \sum_{\substack{w \in N(v) \\ \pos(w) \neq p}} \ell(vw)^\alpha ,
  \end{align*}
 then $v$ is located at $p$ in $L$.
\end{lemma}

\section{Structural Properties of Optimal Layouts}
In this section, we compile some general properties on the structure of optimal layouts, which we use throughout the paper.
One such property is that edges incident to sources never cross.
\begin{lemma}
  \label{lem:properties:no crossing}
  Let $L$ be an optimal valid layout for some instance of \flamecast, let $a$ and $b$ be two sources with different positions, and let $v_a$ and $v_b$ be their respective parents.
  If $v_a \neq v_b$ and $a, b, v_a, v_b$ are not collinear, the edges $av_a$ and $bv_b$ do not intersect.
\end{lemma}
\begin{proof}
  Consider a layout $L$ where two edges $av_a$ and $bv_b$ cross at some point $p$.
  As the load of both edges is $1$, the cost contribution from these two edges is $\cost(av_a) + \cost(bv_b)=\dist(a,v_a) + \dist(b, v_b) = \dist(a,p) + \dist(p,v_a) + \dist(b,p) + \dist(p,v_b)$.

  We construct the alternative layout $L'$ in which $a$ is connected to $v_b$ and $b$ is connected to $v_a$.
  All other connections remain the same.
  In $L'$, the cost contribution of $b v_a$ and $a v_b$ is $\cost(bv_a) + \cost(av_b) = \dist(b,v_a) + \dist(a, v_b)$.
  By the triangle inequality, we know that $\dist(a,v_b) \leq \dist(a,p) + \dist(p,v_b)$ and $\dist(b,v_a) \leq \dist(b,p) + \dist(p,v_a)$.
  This shows that the cost of $L'$ is less than or equal to the cost of $L$.
  The inequality is strict if the points are not collinear.

  Swapping the parents preserves the subtree sizes of both Steiner vertices, and thus, $L'$ is a valid layout.
  Since we can transform any layout with crossing edges into a layout with lower cost without crossing edges, an optimal layout cannot have two edges that are incident to sources cross.
\end{proof}
Moreover, we show that if multiple sources share the same position, then they also share a parent in any cost-optimal layout.
Recall that a layout is not necessarily valid, i.e., the capacities in a cost-optimal layout may be violated.
We first prove the following intermediate result.
\begin{lemma}
  \label{lemma:no-split-op}
  Let $I = (S, T, C, \alpha)$ be an instance of \flamecast with $\alpha \in [0, 1)$, and let $A \subseteq S$ be a set of sources that share the same position.
  Let $L$ be a layout for $I$.
  Let $a_v \in A$ be connected to a Steiner vertex $v$ and $a_w \in A$ be connected to a Steiner vertex $w$ in $L$, where $w \neq v$.
  If $\alpha = 0$, then either reconnecting $a_v$ to $w$ or $a_w$ to $v$ does not increase the cost.
  If $\alpha \in (0, 1)$, then one or the other strictly reduces the cost.
\end{lemma}
\begin{proof}
  For $\alpha = 0$, the cost only depends on the total length of the edges.
  Assuming $\dist(A, v) \leq \dist(A, w)$ without loss of generality, the cost does not increase if $a_w$ is reconnected to $v$.
  Now, let $x_v$ be the lowest ancestor of $a_v$ and $a_w$ if $a_v$ and $a_w$ are in the same tree, and the root of $v$ otherwise.
  For $\alpha \in (0, 1)$, we consider how the cost of the path from $v$ to $x$ changes when adding a child to $v$ or removing a child from $v$.
  Note that for the edges on the path from $x_v$ to its root, the cost does not change since the loads of the ancestors of $x_v$ do not change.
  Figure~\ref{fig:no-split} shows an example before and after reconnecting $a_w$ to $v$.
  \begin{figure}
    \centering
    \includegraphics[page=4]{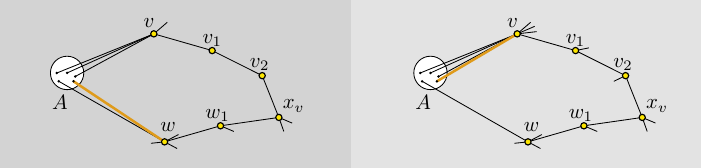}
    \caption{A group of sources that is connected to two distinct vertices, before and after reconnecting $a_w$ to $v$.
    The changed edges are marked in orange.}
    \label{fig:no-split}
  \end{figure}
  Let $P_v = (v = v_0, v_1, \ldots, v_k = x_v)$ be the unique path from $v$ to $x_v$.
  Let $\Delta_v^+$ be the cost increase for the path $P_v$ and the children of $v$ if $a_w$ is connected to $v$.
  Similarly, let $\Delta_v^-$ be the decrease of the cost for the path $P_v$ and the children of $v$ if $a_v$ is removed from $v$.
  We later show that $\Delta_v^+ < \Delta_v^-$, i.e., that the cost increase for adding a child is smaller than the cost gain by removing a child due to the cost parameter $\alpha$.
  For $\alpha \in (0, 1)$, the cost for an edge $e$ (proportional to $\ell(e)^\alpha$) grows slower with increasing load.

  In the case of connecting $a_v$ to $v$, we get an additional edge from $A$ to $v$ with cost $\dist(A, v)$, and the load of each edge $v_{i - 1}v_i$ for $i \in [k]$ along path $P_v$ increases by $1$.
  Then, we obtain an increase of the cost for the path $P_v$ and the children of $v$ of
  \begin{equation*}
    \label{eq:deltaplus}
    \Delta_v^+ = \dist(A,v) + \sum_{i = 1}^{k} \dist(v_{i - 1}, v_i) \cdot \left((\ell(v_{i - 1}) + 1)^\alpha - \ell(v_{i - 1})^\alpha\right) .
  \end{equation*}
  In the case of removing $a_v$ from $v$, we lose an edge between $A$ and $v$ with cost $\dist(A, v)$, and the load of each edge $v_{i - 1}v_i$ for $i \in [k]$ along path $P_v$ decreases by $1$, yielding
  \begin{equation*}
    \label{eq:deltaminus}
    \Delta_v^- = \dist(A,v) + \sum_{i = 1}^{k} \dist(v_{i - 1}, v_i) \cdot \left(\ell(v_{i - 1})^\alpha - (\ell(v_{i - 1}) - 1)^\alpha\right) .
  \end{equation*}
  To show $\Delta_v^+ < \Delta_v^-$, we show the inequality for each $i \in [k]$
  \begin{equation*}
   (\ell(v_{i - 1}) + 1)^\alpha - \ell(v_{i - 1})^\alpha < \ell(v_{i - 1})^\alpha - (\ell(v_{i - 1}) - 1)^\alpha .
  \end{equation*}
  Note that on both sides of the inequality, we have the function $f(z) = (z + 1)^\alpha - z^\alpha$.
  On the left-hand side it is evaluated at $z = \ell(v)$ and on the right-hand side at $z = \ell(v) - 1$.
  As $f(z)$ is strictly decreasing in $z$ for $0 < \alpha < 1$, the inequality is true for each edge, and thus we have $\Delta_v^+ < \Delta_v^-$ in total.

  We define $\Delta_w^+$ and $\Delta_w^-$ analogously.
  Let $x_w$ be the lowest ancestor of $a_v$ and $a_w$ if $a_v$ and $a_w$ are in the same tree, and the root of $w$ otherwise.
  Furthermore, let $P_w = (w = w_0, w_1, \ldots, w_k = x_w)$ be the unique path from $w$ to $x_w$.
  The increase of the cost for the path $P_w$ and the children of $w$ if $a_v$ is connected to $w$ is denoted by $\Delta_w^+$, and $\Delta_w^-$ is the decrease of the cost for the path $P_w$ and the children of $w$ if $a_w$ is removed from $w$.
  Using the same reasoning as for $v$, we show that $\Delta_w^+ < \Delta_w^-$.

  We can now express the total cost change when reconnecting $a_w$ to $v$ in terms of $\Delta_v^+$ and $\Delta_w^-$ as $\Delta_v^+ - \Delta_w^-$.
  Analogously, reconnecting $a_v$ to $w$ changes the cost by $\Delta_w^+ - \Delta_v^-$.
  We show that one of the changes decreases the total cost by summing up the cost changes for both modifications.
  Rearranging the sum and using $\Delta_v^+ < \Delta_v^-$ and $\Delta_w^+ < \Delta_w^-$ gives us
  \begin{equation*}
    (\Delta_v^+ - \Delta_w^-) + (\Delta_w^+ - \Delta_v^-) = (\Delta_v^+ - \Delta_v^-) + (\Delta_w^+ - \Delta_w^-) < 0 .
  \end{equation*}
  This directly implies that either $\Delta_v^+ - \Delta_w^-$ or $\Delta_w^+ - \Delta_v^-$ is less than $0$, and thus, one of the changes reduces the total cost.
\end{proof}
The argument above can be repeated for all sources in $A$ that do not share a parent.
Since the cost is only reduced in every step for $\alpha \in (0,1)$, this process does not cycle, and thus, all sources in $A$ share the same parent in an optimal layout.
For $\alpha = 0$, the cost does not increase if all sources in $A$ share the same parent.
Repeating this for every set of sources that share a position gives us the following corollary.
\begin{corollary}
  \label{cor:no-split}
  In an optimal layout for an instance of \flamecast with $\alpha \in (0,1)$, sources with the same position also have the same parent.
  For $\alpha = 0$, there is an optimal layout where sources with the same position also have the same parent.
\end{corollary}

\section{Complexity}
With its resemblance to the Steiner tree problem, it should come as no surprise that \flamecast in general is NP-hard.
In this section, we take a more fine-grained look at its computational complexity, depending on the number of intermediate layers $\lambda$ and on the parameter $\alpha$ controlling the edge cost.
We first show that if $\lambda = 0$ or $\alpha = 1$, then the problem is solvable in polynomial time.
However, for at least one intermediate layer, we prove NP-hardness for every $\alpha < 1$ by reducing from \textsc{Planar Monotone 3-SAT}.

\subsection{No Intermediate Layers or \texorpdfstring{\boldmath$\alpha = 1$}{𝛼 = 1}}
\label{sec:two layer}
In the case $\lambda = 0$, there are no intermediate layers but only sources and sinks.
Thus, for a given instance $I$, we only have to decide for each source which sink to connect it to.
With this, finding a minimum-cost assignment is essentially a matching problem.
More formally, we model the problem as follows.
We construct a complete bipartite graph with parts $V_S$ and $V_T$ where $V_S$ is the set of sources.
In $V_T$, we have $c_0$ copies of each sink, where $c_0$ is the capacity of the sinks.
The weight of an edge between some vertex in $V_S$ and some vertex in $V_T$ is equal to the distance between the corresponding source and the corresponding sink (recall that each source and each sink has a fixed position in $I$).

The edges in a maximum matching with minimum weight in this bipartite graph directly correspond to an optimal assignment.
Note that no capacities are violated since at most $c_0$ sources are assigned to each sink by construction.
If not all sources are assigned to a sink, then the total capacity of the sinks is not sufficient and there is no valid solution.
A maximum matching with minimum weight can be determined in time $O((c_0 \cdot |T|)^3)$ using the Kuhn-Munkres algorithm~\cite{kuhn1955,munkres1957,tomizawa1971}.

For $\alpha = 1$ and any number of layers, we first show that bundling sources with intermediate layers cannot reduce the cost of a layout.
\begin{observation}
  Let $I$ be a feasible instance of \flamecast with $\alpha = 1$.
  There is an optimal valid layout for $I$ where every source is directly connected to a sink.
\end{observation}
\begin{proof}
  Let $L$ be a layout for $I$.
  If $L$ only consists of sources and sinks, we are done.
  Otherwise, there is a Steiner vertex $v$.
  For the parent $u$ of $v$, it is $\ell(uv)^\alpha = \sum_{w \in N(v) \setminus \set{u}} \ell(vw)^\alpha$ for $\alpha = 1$, i.e., $u$ weighs as much as all other neighbors together.
  By Lemma~\ref{lem:properties:majority-weber-point}, it is optimal if $v$ has the same position as $u$.
  Thus, we can remove $v$ and connect all children of $v$ to $u$ without increasing the cost.
  Repeating this argument eliminates all Steiner vertices.
\end{proof}
The previous lemma directly implies that $\alpha = 1$ is just a special case of having no intermediate layers.
Thus, it can also be solved using a matching algorithm.
We obtain the following theorem.
\begin{theorem}
  \label{thm:poly-case}
    If $\alpha = 1$ or if there are no intermediate layers, then \flamecast is solvable in $O((c_0 \cdot |T|)^3)$ where $|T|$ is the number of sinks and $c_0$ is their capacity.
\end{theorem}
Note that we can assume that $c_0 \leq |S|$ since the load of a vertex is at most the number of all sources.
Thus, the proposed algorithm is indeed polynomial in the input size.

\subsection{One Intermediate Layer and \texorpdfstring{\boldmath$\alpha < 1$}{𝛼 < 1}}

Here we contrast the results from the previous section by showing that \flamecast is NP-hard, even if we have only one intermediate layer and for any $\alpha < 1$.
We show this with a reduction from the NP-hard problem \textsc{Planar Monotone 3-SAT}~\cite{deberg2010}.
Moreover, we show that approximating \flamecast within a factor of $1 + \frac{1 + 6^\alpha - 7^\alpha}{14^\alpha n}$ is NP-hard, where $n$ is the number of sources.

\subparagraph*{Planar Monotone 3-SAT}

Let $x_1, \dots, x_n$ be a set of Boolean variables.
The appearance of a variable $x_i$ in a Boolean formula is called a \emph{literal}, and it can be \emph{positive} or \emph{negative} if it appears as $x_i$ or in its negated form $\overline{x_i}$, respectively.
A disjunction of literals is called a \emph{clause}.
The \emph{size} of a clause is the number of literals it contains, e.g., $x_1 \vee \overline{x_2} \vee x_3$ is a clause of size~3.
A Boolean formula $\Phi$ is in \emph{conjunctive normal form (CNF)} if it is a conjunction of clauses, i.e., $\Phi = c_1 \wedge \dots \wedge c_m$ for clauses $c_1, \dots, c_m$.
The problem \textsc{3-SAT} has as input a CNF formula $\Phi$ where each clause has size at most~3 and one has to decide whether $\Phi$ has a satisfying assignment, i.e., whether one can assign values true and false to the variables such that $\Phi$ evaluates to true.

The variant of \textsc{3-SAT} we reduce from poses some additional restrictions to the inputs, which we introduce in the following.
A clause is called \emph{monotone} if it only contains positive or only negative literals.
We then also call the clause itself \emph{positive} or \emph{negative}, respectively.
A CNF-formula $\Phi$ is \emph{monotone} if all its clauses are monotone.

Consider an instance $\Phi$ of \textsc{3-SAT} with variable set $X$ and clause set $C$.
The \emph{clause--variable graph} is the bipartite graph with vertex set $X \cup C$ and an edge between $x \in X$ and $c \in C$ if and only if the variable $x$ appears in the clause $c$.
A \emph{rectilinear planar drawing} of the clause--variable graph is an intersection-free representation, where each variable is a line segment on the $x$-axis and every clause has the shape of an \textsf{E} rotated by $90^\circ$; also see Figure~\ref{fig:sat}.
Thus, a clause consists of one horizontal segment with three vertical segments attached to it.
With the \emph{position of the clause}, we refer to the point where the horizontal segment and the middle vertical segment meet.
Removing this point from the drawing of the clause decomposes it into three pieces, which we call the \emph{legs} of the clause (one of which is just a segment and the other two have one bend each).
It is known that every instance of \textsc{3-SAT} with planar clause--variable graph has a rectilinear planar drawing~\cite{knuth1992}.

The problem \textsc{Planar Monotone 3-SAT} has as input a monotone \textsc{3-SAT} formula, together with a rectilinear planar drawing of its clause--variable graph such that all positive clauses are above and all negative clauses are below the $x$-axis; see Figure~\ref{fig:sat}.
We additionally assume that at each variable, each vertical segment belonging to a negative occurrence is strictly to the left of each vertical segment belonging to a positive occurrence.
Since the positive and the negative clauses can be arranged independently from each other, this is not a real restriction.
De Berg and Khosravi showed that this variant of \textsc{3-SAT} remains NP-hard \cite{deberg2010}.
We can assume that the given drawing is a grid drawing of polynomial size, i.e., the endpoints of each segment have integer coordinates whose magnitude is polynomial in the number of clauses and variables.
By scaling such a drawing with a factor of $6$, we can additionally assume that all distances between non-adjacent line-segments are at least $6$ and that all coordinates are even integers.

\begin{figure}
  \centering
  \includegraphics[page=3]{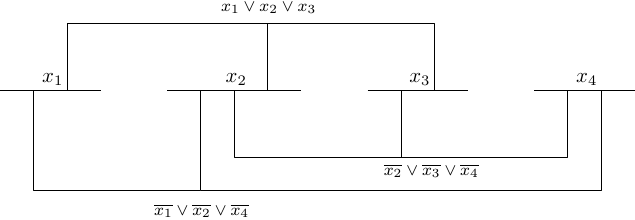}
  \caption{A rectilinear drawing of the monotone 3-SAT instance $I = (\set{x_1, x_2, x_3, x_4}, \set{x_1 \lor x_2 \lor x_3, \overline{x_1} \lor \overline{x_2} \lor \overline{x_4}, \overline{x_2} \lor \overline{x_3} \lor \overline{x_4}})$.
  }
  \label{fig:sat}
\end{figure}
%



\subparagraph*{Reduction From \textsc{3-SAT} to \flamecast}

Given an instance of \textsc{Planar Monotone 3-SAT} consisting of a formula $\Phi$ together with a drawing as described above, we now construct an instance $I_\Phi$ of \flamecast mimicking the drawing of $\Phi$, such that $\Phi$ is satisfiable if and only if $I_\Phi$ is a yes-instance.
For this, we model a variable gadget comprising multiple literal gadgets that behaves like a variable with a binary decision.
Additionally, we have a clause gadget that is satisfied if at least one of the literal gadgets is in the correct configuration.
Moreover, we need transport gadgets to propagate the decisions from the variables to the clauses.

Before describing these gadgets, let us first fix some general properties of the resulting \flamecast instance $I_\Phi$.
In $I_\Phi$ we allow one intermediate layer between the sinks and the sources.
Moreover, the sources of $I_\Phi$ will form groups with $g$ or $2g$ sources at the same position, where $g \geq 7$ is an arbitrarily chosen integer.
For the sake of understanding the general idea of our construction, it is useful to assume for now that intermediate vertices are always placed at the same position as one of these source groups, bundling them together before connecting to a sink.

The basic building blocks of our construction are shown in Figure~\ref{fig:gadgets}.
Each of these building blocks has one sink with capacity~$2g$.
When we later refer to the position of one of these building blocks, we mean the position of its sink.
Additionally, each building block has one, two, or three source groups on its boundary at distance~$1$ from the sink.
Each source group is composed of $g$ or $2g$ sinks as indicated by the numbers in Figure~\ref{fig:gadgets}.
We note that for the line and end gadgets, we have two variants: one of the source groups can have size $g$ or $2g$.

We combine these building blocks by placing them next to each other such that they share one of the source groups on their boundary.
In our construction, each source group will be shared by exactly two building blocks and finding a solution to the \flamecast instance will boil down to deciding to which of the two corresponding sinks to connect.
The role of the different building blocks is the following.

The \emph{clause gadget} represents the requirements posed by a clause.
Note that the fact that sinks have capacity~$2g$ means that only two of the three source groups on its boundary can connect to the sink of the clause gadget.
This corresponds to the restriction that at most two of the literals of a clause in $\Phi$ can be false.

The \emph{literal gadget} represents the value of a literal.
There are two canonical configurations in which the literal gadget can be.
The first option is that the source group of size $2g$ connects to the sink of the literal gadget and the two source groups of size $g$ connect to the sink of an adjacent gadget.
The second option is the direct opposite, i.e., that the two smaller groups connect to the sink of the literal gadget and the large group connects to the sink of an adjacent gadget.
We note that at this point it is not obvious that these are the only two configurations we really need to consider, but this will become clear later.

The \emph{corner} and the \emph{line} gadgets have the purpose of synchronizing the decisions between the literal gadgets of the same variable and of transporting the decisions from the literal gadgets to the corresponding clause gadgets.
For both gadgets, the canonical configuration in a solution of the \flamecast instance $I_\Phi$ is that exactly one of the two source groups connects to the sink of the gadget.
The purpose of the \emph{end gadget} is to collect loose ends.

\begin{figure}
  \centering
  \includegraphics{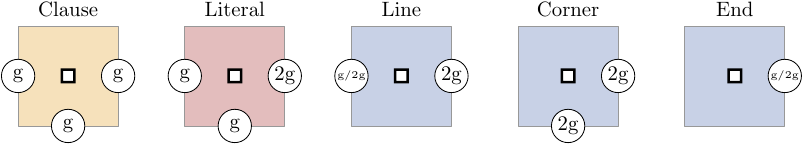}
  \caption{All gadgets that are used in the reduction.
  The central square represents a sink with capacity $2g$.
  Each group of sources is represented by a circle and have distance $1$ to the sink.
  The size of each source group inside the circle.
  For the line gadget and the end gadget, there are two variants where one source group can either have size $g$ or $2g$.}
  \label{fig:gadgets}
\end{figure}

For the synchronization between the literal gadgets, we combine them into a \emph{variable gadget} as shown in Figure~\ref{fig:variable}.
To describe this formally, consider a variable $x$ in $\Phi$.
In the drawing of $\Phi$, $x$ is represented as a horizontal line between points $p_{\min}$ and $p_{\max}$.
Moreover, there are vertical segments attaching to this horizontal line.
Let $p$ be a point (between $p_{\min}$ and $p_{\max}$) where such a vertical segment attaches.
Then we place a literal gadget at $p$, rotated by $180^\circ$ if the vertical segment attaches from above, i.e., if it represents a positive literal.
Recall that in the given drawing, the vertical segments representing negative literals are sorted to the left of the vertical segments representing positive literals, and thus, the literal gadgets are sorted as well.
Moreover, at every point on the horizontal line of $x$ between $p_{\min}$ and $p_{\max}$ with even $x$-coordinate that is not already occupied by a literal gadget, we place a (potentially rotated) line gadget.
The size of the source groups of these line gadgets is chosen to be $2g$ whenever possible, i.e., except in cases where the line gadget shares its boundary with a literal gadget that has a source group of size $g$ on the corresponding side.
Recall that we scaled the drawing of $\Phi$ by a factor of $6$.
Thus, all coordinates are even and this construction nicely works out such that two consecutive gadgets on the line of $x$ share a boundary.
Moreover, no line gadget is adjacent to two literal gadgets, i.e., a line gadget always has at least one source group of size~$2g$.
We finalize the variable gadget by adding two end gadgets with distance $2$ to the left of $p_{\min}$ and to the right of $p_{\max}$ (the latter being rotated by $180^\circ$).
Note that, again due to the scaling by a factor of $6$, no two building blocks of different variable gadgets overlap.

\begin{figure}
  \centering
  \includegraphics[page=2]{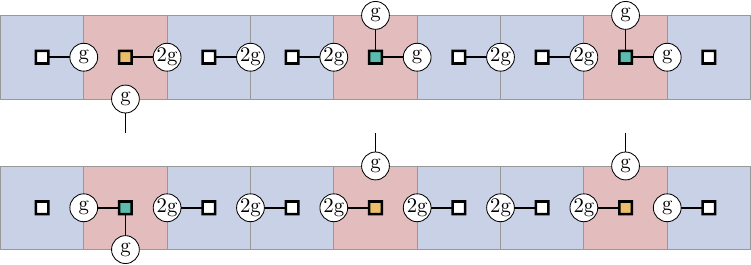}
  \caption{A variable gadget in its two canonical configurations.
  The configuration above is true whereas the configuration below is false.
  The color of a literal gadget sink indicates the configuration of the literal gadget: green for true, orange for false.}
  \label{fig:variable}
\end{figure}

For each clause of $\Phi$, we simply place a clause gadget at the position of the clause in the drawing of $\Phi$.
Moreover, we connect the clause gadgets with their corresponding literal gadgets using a sequence of line gadgets following the drawings of the legs.
In case the leg has a bend, we place a corner gadget (with the appropriate orientation) at the position of the bend.
The full construction can be seen in Figure~\ref{fig:instance}.

As the drawing of $\Phi$ has polynomial size, the instance $I_\Phi$ has polynomial size and can be constructed in polynomial time.

In the following, we define a specific type of layout, namely canonical layouts.
We show that $\Phi$ is satisfiable if and only if there is a valid canonical layout for $I_\Phi$.
Moreover, we show that any optimal layout for $I_\Phi$ is canonical.

\begin{figure}
  \centering
  \includegraphics[angle=90, page=2]{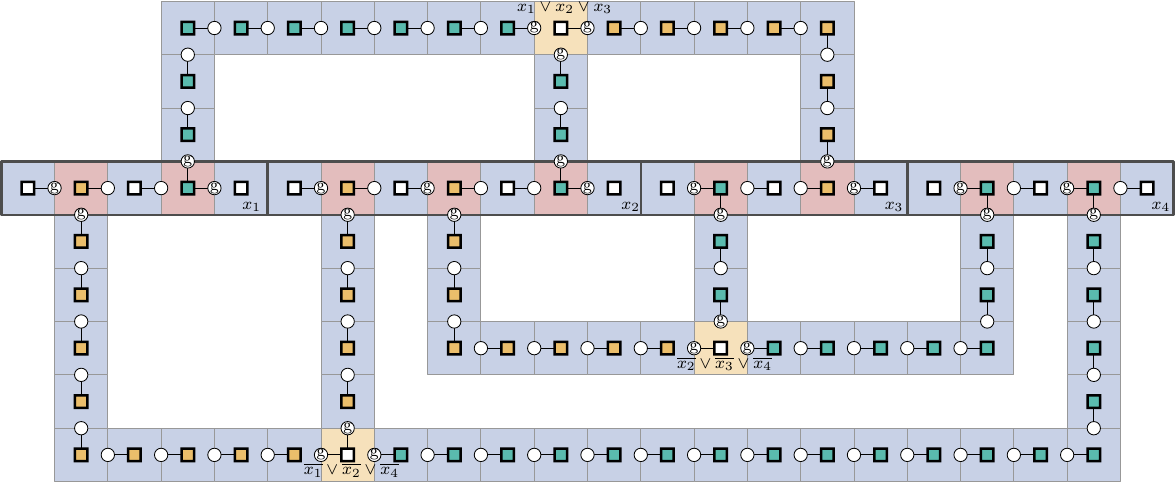}
  \caption{The constructed instance $I'$ for instance $I = (\set{x_1, x_2, x_3, x_4}, \set{x_1 \lor x_2 \lor x_3, \overline{x_1} \lor \overline{x_2} \lor \overline{x_4}, \overline{x_2} \lor \overline{x_3} \lor \overline{x_4}})$.
  In the unlabeled groups, there are $2g$ sources.
  A possible solution is setting $x_1$ and $x_2$ to true and $x_3$ and $x_4$ to false.
  The colors indicate the configuration of a gadget, green for true and orange for false.
  Note that distances between non-adjacent segments may be smaller than required for better readability.}
  \label{fig:instance}
\end{figure}



\subparagraph*{Canonical Layout}

In the description of our construction of $I_\Phi$ above, we sometimes made the implicit assumption that each source group is connected to a distinct Steiner vertex that is placed at the same position as the group and that each Steiner vertex is connected to a sink at distance $1$.
We call any layout that satisfies this property a \emph{canonical} layout.
We note that this definition includes layouts that are no solutions to $I_\Phi$, i.e., a canonical layout may violate the sink capacities.

Note that connecting a source group of size $k$ to a sink at distance $1$ costs $k^\alpha \cdot 1$.
Thus, every canonical layout of $I_\Phi$ has the same cost, which we refer to as \emph{canonical cost}.

\begin{observation}
  \label{lemma:canonical-cost}
  Let $n_g$ and $n_{2g}$ be the number of groups of size $g$ and $2g$ in the \flamecast instance $I_\Phi$, respectively.
  Then the cost of every canonical layout if $I_\Phi$ is $g^\alpha \cdot n_{g} + (2g)^\alpha \cdot n_{2g}$.
\end{observation}

\subparagraph*{Correctness of the Reduction}
We show that the formula $\Phi$ has a satisfying assignment if and only if the \flamecast instance $I_\Phi$ has a solution of cost $g^\alpha \cdot n_{g} + (2g)^\alpha \cdot n_{2g}$ (i.e., the canonical cost).
The argument goes as follows.
We show that any non-canonical layout has cost strictly larger than the canonical cost.
Thus, for $I_\Phi$ to have a solution of the required cost, it must have a valid canonical layout.
It then remains to show that $\Phi$ has a satisfying assignment if and only if $I_\Phi$ has a valid canonical layout.
We start with the latter.

\begin{lemma}
  \label{lemma:np-hard-correctness}
  The \textsc{Planar Monotone 3-SAT} formula $\Phi$ is satisfiable if and only if the \flamecast instance $I_\Phi$ has a valid canonical layout.
\end{lemma}
\begin{proof}
  For each literal gadget, there are two possibilities: either the source group of size $2g$ is connected to the sink, which we call the false configuration, or not, which we call the true configuration.
  We first show that this is consistent for all literal gadgets in a variable gadget in the sense that is impossible that both a positive and a negative literal are in their true configuration.

  Assume there is a literal gadget representing a negative occurrence that is in its true configuration, i.e., the group of size $2g$ is not connected to the sink.
  Then, this group must be connected to the sink to its right.
  By construction, no two groups on the $x$-axis can be connected to the same sink, thus, all groups on the $x$-axis to the right of the negative literal gadget we consider are also connected to their sink on the right.
  In particular, this is true for all positive literal gadgets, which are to the right of negative literal gadgets by construction.
  This means that for these, the group of size $2g$ is connected to the sink of the corresponding literal gadget, and thus, each positive literal gadget is in its false configuration.
  Conversely, if a positive literal gadget is in its true configuration, we argue analogously that all negative literal gadgets are in their false configuration.

  Thus, we also get two possibilities for a variable gadget: either all positive literal gadgets are false, which we call the false configuration, or not, which we call the true configuration.

  If $\Phi$ is a yes-instance, then there is a truth assignment of the variables such that each clause evaluates to true.
  For each literal gadget, we choose the corresponding configuration.
  The decision at each literal gadget is propagated to the clause gadget by transport gadgets, starting at the literal gadgets:
  if the sink we currently consider is unused, then we connect the following source group to the sink.
  Otherwise, we connect the following source group to the next sink.
  Since the sink of each transport gadget is connected to exactly one source group, the capacity constraints at each transport gadget are satisfied.
  At a clause gadget, an adjacent source group is not connected to the sink if and only if the corresponding literal evaluates to true.
  Since we have a satisfying assignment of $\Phi$, there is at least one source group at each clause gadget that is not connected to the sink, and thus, the sink capacity at each clause gadget is respected.

  Assume now that there is a valid canonical layout for the \flamecast instance $I_\Phi$.
  We set a variable to true if the corresponding variable gadget is in its true configuration and false otherwise.
  We show that each clause in $\Phi$ contains a literal that evaluates to true.
  In each clause gadget, there is a source group that is not connected to the sink due to the limited capacity of the sink.
  Thus, this source group must be connected to the other adjacent gadget.
  Since only one source group may be connected to a sink of a transport gadget, starting at the clause gadget, we can follow the chain back to the variable gadget, where each source group is connected to its previous sink.
  At the variable gadget, the corresponding source group of the literal gadget is connected to its sink which means that the corresponding literal evaluates to true.
  If the literal is positive, then the corresponding variable is set to true, and otherwise to false.
  Thus, such a truth assignment is a solution for the \textsc{Planar Monotone 3-SAT} instance $I$.
\end{proof}

It remains to show that any non-canonical layout has cost strictly larger than the canonical cost.
Note that this is a statement about layouts in general, including layouts where capacity constraints may be violated.
The idea is to start with a non-canonical layout and to iteratively modify it to make it ``more canonical'' until we reach a canonical layout.
We show that after each step, the cost of the layout does not increase, and there is a step that strictly reduces the cost.
In the following, we first list the possible modifications.
%
\subparagraph{Operations making a layout canonical}
\label{par:operations}
To make a non-canonical layout canonical, we apply the following operations in the specified order.
More specifically, each type of operation is applied exhaustively before moving on to the next type of operation.
\begin{description}
    \item[Operation 1]\label{op:low-dist-to-intermediate}Let $s$ be a source that is connected to a Steiner vertex $v$ with $\dist(s, v) \geq 1$.
    Add a new Steiner vertex $v_s$ at the location of $s$, reconnect $s$ to $v_s$, and connect $s$ to a sink $t$ that has distance $1$ to $v_s$.
    \item[Operation 2]\label{op:distance-to-sink}Let $v$ be a Steiner vertex connected to a sink $t$, where $t$ has distance more than $1$ to at least one child of $v$.
    If there is a sink $t'$ that has distance at most $1$ to every child of $v$, reconnect $v$ to $t'$.
    \item[Operation 3]\label{op:unified}Let $A$ be a source group, where a subset $A_v \subseteq A$ is connected to a Steiner vertex $v$ and a subset $A_w \subseteq A$ is connected to a Steiner vertex $w \neq v$.
    Either reconnect each source in $A_v$ to $w$ or each source in $A_w$ to $v$.
    \item[Operation 4]\label{op:no-shared-intermediate}Let $v$ be a Steiner vertex connected to two distinct source groups $A$ and $B$.
    Add a Steiner vertex $v_A$ at the location of $A$, connect $v_A$ to a sink of distance $1$ and reconnect $A$ to $v_A$.
    Repeat for $B$, adding a Steiner vertex $v_B$ at the location of $B$.
  \end{description}
\begin{figure}
  \centering
  \includegraphics{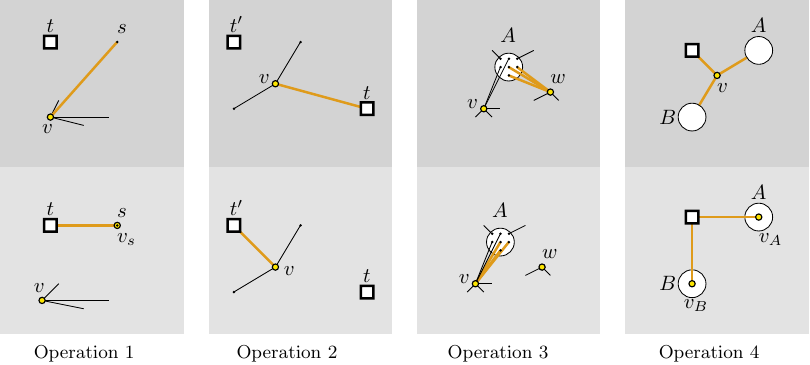}
  \caption{Each operation that is used to make a non-canonical layout canonical.
  The darker area depicts the state before the operation, the lighter area shows the state after the operation.
  The changed edges are marked in orange for each operation.}
  \label{fig:operations}
\end{figure}
Each operation can be seen in Figure~\ref{fig:operations}.
By proving the following properties after each step, we first observe that the resulting layout is indeed canonical.
\begin{observation}
    \label{obs:canonical-after-ops}
    Let $L$ be a non-canonical layout of $I_\Phi$.
    Exhaustively applying Operations~$1$ to $4$ in that order to $L$ yields a canonical layout.
    Moreover, let $L_1, L_2, L_3, L_4$ be a sequence of layouts, where $L_i$ is a resulting layout after having exhaustively applied Operations~$1$ to~$i$ to $L$.
    For each $L_i$, Properties~$1$ to~$i$ hold, which are the following.
\begin{enumerate}
  \item \label{item:low-dist-to-intermediate}The distance between a source and its parent is at most $1$.
  Moreover, children of a Steiner vertex have pairwise distance at most $\sqrt{2}$,
  and each Steiner vertex is connected to sources of at most two distinct source groups.
  \item \label{item:distance-to-sink}Every source is  connected to a sink at distance~$1$.
  \item \label{item:unified}Sources of the same group share a parent.
  \item \label{item:no-shared-intermediate}All children of a Steiner vertex belong to the same source group.
\end{enumerate}
\end{observation}
\begin{proof}
    After exhaustively applying Operation~1, it is clear that the distance between a source and its parent is less than $1$.
    As a consequence, no two sources with distance at least $2$ have the same parent since one of the groups would have distance at least $1$ to its parent by the triangle inequality.
    It directly follows from our construction of $I_\Phi$ that if source groups have distance less than $2$ from each other, then the distance between them is $\sqrt{2}$.
    Furthermore, there are no three source groups with pairwise distance less than~$2$.
    Thus, Property~\ref{item:low-dist-to-intermediate} holds for $L_1$.

    We now consider the layout after exhaustively applying Operation~2.
    For sources that do not share their parent with sources of another group, it is obvious that Property~\ref{item:distance-to-sink} holds.
    Any other Steiner vertex $v$ is connected to sources of exactly two source groups $A$ and $B$ that have distance $\sqrt{2}$ by Property~\ref{item:low-dist-to-intermediate}.
    By construction, $A$ and $B$ are part of the same literal, clause or corner gadget.
    Thus, the corresponding sink that has distance $1$ to both $A$ and $B$, and Property~\ref{item:distance-to-sink} holds for all Steiner vertices.

    After exhaustively applying Operation~3, sources of the same group trivially share a parent, i.e., Property~\ref{item:unified} holds.
    Operation~4 splits any pair of source groups with distance $\sqrt{2}$ that have the same parent.
    In the resulting layout, every group is connected to a sink at distance $1$ via a distinct Steiner vertex, which is captured by Property~\ref{item:no-shared-intermediate}.
    In other words, the resulting layout is canonical.
\end{proof}

It remains to show that applying all possible operations in the specified order results in a layout with lower cost.
We first prove the following technical lemma, which can later be used to compare the cost of certain non-canonical structures to the cost after they are made canonical.
\begin{lemma}
  \label{lemma:cost-difference}
  Let $a, b \in \mathbb{N}$ with $b \geq a$.
  Let
  \begin{equation*}
    f(\alpha, (d_t, d_a, d_b)) = (a + b)^\alpha d_t + ad_a + bd_b - a^\alpha - b ^\alpha .
  \end{equation*}
  Then it is
    \begin{align*}
    &f(\alpha, (d_t, d_a, d_b)) \\
    &\geq \min\left(f(0, (d_t, d_a, d_b)), f\left(1, \left(\frac{\log(a)a + \log(b)b}{\log(a + b)(a+b)}, d_a, d_b\right)\right), f(1, (d_t, d_a, d_b))\right)
  \end{align*}
  for each $d_t, d_a, d_b \in \mathbb{R}_{\geq 0}$ and $\alpha \in [0, 1)$.
\end{lemma}

\begin{proof}
  Let $x = (d_t, d_a, d_b) \in \mathbb{R}^3_{\geq 0}$ be arbitrary but fixed.
  Then, $f(\alpha, x)$ is a function in $\alpha$.
  We consider values for $\alpha$ where $f(\alpha, x)$ may be minimal.
  There are two types of minima: minima at boundaries or local minima.
  The function is always at least the value of the smallest minimum on the interval $\alpha \in [0, 1]$.
  Moreover, we consider local minima of $f(\alpha, x)$ by determining the derivative $f'$ with respect to $\alpha$, which is
  \begin{equation*}
    f'(\alpha, x) = \log(a + b)(a+b)^\alpha d_t - \log(a)a^\alpha - \log(b)b^\alpha .
  \end{equation*}
  By simple rearranging, we see that this is equal to $0$ if and only if
  \begin{equation*}
        d_t = \frac{\log(a)a^\alpha + \log(b)b^\alpha}{\log(a + b)(a + b)^\alpha}.
  \end{equation*}
  We capture the right side of the equation by the function $g(\alpha)$.
  Then, at any local minimum $\alpha^*$, it is $g(\alpha^*) = d_t$.
  For such an $\alpha^*$ with $g(\alpha^*) = d_t$, the value of $f(\alpha^*, d_t, d_a, d_b)$ is
  \begin{align*}
    f(\alpha^*, d_t, d_a, d_b) &= (a + b)^{\alpha^*} g(\alpha^*) + a d_a + b d_b - a^{\alpha^*} - b^{\alpha^*} \\
    &= (a + b)^{\alpha^*} \frac{\log(a)a^{\alpha^*} + \log(b)b^{\alpha^*}}{\log(a + b)(a + b)^{\alpha^*}} + a d_a + b d_b - a^{\alpha^*} - b^{\alpha^*} \\
    &= \frac{\log(a)}{\log(a + b)} a^{\alpha^*} + \frac{\log(b)}{\log(a + b)} b^{\alpha^*} + a d_a + b d_b - a^{\alpha^*} - b^{\alpha^*} \\
    &= \left(\frac{\log(a)}{\log(a + b)} - 1\right) a^{\alpha^*} + \left(\frac{\log(b)}{\log(a + b)} - 1\right) b^{\alpha^*} + a d_a + b d_b \\
    &> \left(\frac{\log(a)}{\log(a + b)} - 1\right) a + \left(\frac{\log(b)}{\log(a + b)} - 1\right) b + a d_a + b d_b \\
    &= \frac{\log(a)a + \log(b)b}{\log(a + b)} + a d_a + b d_b - a - b \\
    &= (a + b)\frac{\log(a)a + \log(b)b}{\log(a + b)(a + b)} + a d_a + b d_b - a - b \\
    &= f\left(1, \left(\frac{\log(a)a + \log(b)b}{\log(a + b)(a+b)}, d_a, d_b\right)\right) .
  \end{align*}
\end{proof}
The following corollary makes the previous lemma usable for comparing costs of different layouts by choosing suitable values for the variables in the previous lemma.
The setting is shown in Figure~\ref{fig:cost-difference}.
    \begin{figure}
    \centering
    \includegraphics[page=7]{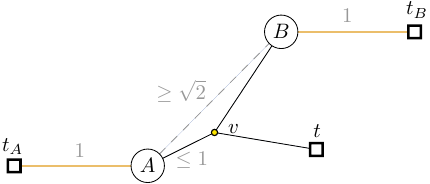}
    \caption{The setting of \cref{cor:cost-difference}. The orange edges mark the alternative layout $L^*$.}
    \label{fig:cost-difference}
  \end{figure}
\begin{corollary}
  \label{cor:cost-difference}
  Let $I$ be an instance of \flamecast with one intermediate layer and $\alpha \in [0, 1)$.
  Let $A$ and $B$ be two source groups with $\dist(A, B) \geq \sqrt{2}$, $|B| \geq |A|$, and $|B| \geq 7$.
  Let $t_A$ and $t_B$ be sinks with $\dist(A, t_A) = 1$ and $\dist(B, t_B) = 1$.
  Furthermore, let $L$ be a layout where $A$ and $B$ are connected via the same Steiner vertex $v$ to a sink $t$ with $\dist(A, v) \leq 1$.
  Let $L^*$ be a modified layout of $L$, where $A$ is connected to $t_A$ and $B$ is connected to $t_B$ using a distinct Steiner vertex for each group.
  Then, it is $\cost(L) - \cost(L^*) \geq \min\left(1, \sqrt{2}|A| + \dist(B, t) - 2, |A|(\dist(A, t) - 1) + |B|(\dist(B, t) - 1)\right)$.
\end{corollary}
\begin{proof}
  Let $a = |A|$ and $b = |B|$, and let $d_t = \dist(v, t)$, $d_a = \dist(A, v)$ and $d_b = \dist(B, v)$ for shorter notation.
  Then, the cost for embedding the edges incident to $v$ is $c(\alpha) = (a + b)^\alpha d_t + ad_a + bd_b$.
  The cost for connecting $A$ to $t_A$ and $B$ to $t_B$ is $c^*(\alpha) = a^\alpha + b^\alpha$.
  The difference $c(\alpha) - c^*(\alpha)$ is equal to $f(\alpha, (d_t, d_a, d_b))$ in Lemma~\ref{lemma:cost-difference}.

  By Lemma~\ref{lemma:cost-difference}, one of three specified values is a lower bound for the cost difference.
  We differ between the three possible cases.
  In the first case, it is
  \begin{align*}
    f(\alpha, (d_t, d_a, d_b)) &\geq f(0, (d_t, d_a, d_b)) \\
    &= d_t + ad_a + bd_b - 2 \\
    &\geq (d_t + d_b) + a (d_a + d_b) + (b - a - 1) d_b - 2 .
  \end{align*}
  We use the triangle inequality to obtain lower bounds.
  If $a = b$, then this is at least $b \dist(A, B) - 2$.
  If $a < b$, then this is at least $a\dist(A, B) + \dist(B, t) - 2 \geq \sqrt{2} a + \dist(B, t) - 2$.

  In the second case, it is
  \begin{align*}
    f(\alpha, (d_t, d_a, d_b)) &\geq f(1, (d_t, d_a, d_b)) \\
    &= (a + b) d_t + a d_a + b d_b - a - b \\
    &\geq a(d_a + d_t - 1) + b(d_b + d_t - 1) \\
    &\geq a(\dist(A, t) - 1) + b(\dist(B, t) - 1) ,
  \end{align*}
  where the last line again follows from triangle inequality.

  In the last case, it is
  \begin{align*}
    f(\alpha, (d_t, d_a, d_b)) &\geq f\left(1, \left(\frac{\log(a)a + \log(b)b}{\log(a + b)(a+b)}, d_a, d_b\right)\right) \\
    &\geq \frac{\log(a)a + \log(b)b}{\log(a + b)} + a d_a + b d_b - a - b \\
    &\geq \frac{\log(b)b}{\log(a + b)} + a (d_a + d_b) + (b - a) d_b - a - b \\
    \intertext{By the triangle inequality, it is $d_a + d_b \geq \dist(A, B)$. Since $d_a \leq 1$, it is $d_b \geq \dist(A, B) - 1$. Thus,}
    &\geq \frac{\log(b)b}{\log(a + b)} + a \dist(A, B) + (b - a) (\dist(A, B) - 1) - a - b \\
    &= \frac{\log(b)b}{\log(a + b)} + a (\dist(A, B) - \dist(A, B) + 1 - 1) + b (\dist(A, B) - 2) \\
    &= b \left(\dist(A, B) + \frac{\log(b)}{\log(a + b)} - 2\right) \\
    &\geq b \left(\dist(A, B) + \frac{\log(b)}{\log(2 b)} - 2\right) \\
    &\geq 7 \left(\sqrt{2} + \frac{\log(7)}{\log(14)} - 2\right) \\
    &\geq 1 .
  \end{align*}
  Note that $\frac{\log(b)}{\log(2b)}$ is increasing for larger $b$.
  For $b \geq 3$, the last case dominates the lower bound $b\dist(A, B) - 2$ from the first case.
\end{proof}
We are now ready to prove that every cost-optimal layout is canonical.
In fact, we prove an even stronger statement, namely that the gap between the canonical cost and the cost of a valid layout for $I_\Phi$ if $\Phi$ is a no-instance is at least $1 + (g-1)^\alpha - g^\alpha$.
From Lemma~\ref{lemma:np-hard-correctness}, we know that any valid layout of a no-instance is non-canonical.
Roughly, we consider the sequence of operations that makes a non-canonical layout canonical as in Observation~\ref{obs:canonical-after-ops}.
In particular, we show that no operation increases the cost and analyze how much the cost decreases in the very last step.

To simplify the analysis, we can state even stronger properties for how a valid layout of a no-instance may look.
In particular, it is not only non-canonical itself, but it is also impossible to make it a valid canonical layout by applying the specified operations.
This property rules out layouts where only Operation~4 may be applied since applying Operation~4 does not affect the validity of the layout.
Thus, instead of considering the cost of all possible valid non-canonical layouts, we may restrict ourselves to valid non-canonical layouts where Operations~1, 2, or 3 can be applied.

We capture this by broadening the notion of canonical layouts.
In the following, we say that a layout is canonical if all sources of a group share the same parent and each group is assigned a sink at distance $1$.
Note that this property is ensured after applying Operation~3 exhaustively.
A canonical layout is \emph{normalized canonical} if additionally, no two source groups share the same parent; otherwise it is \emph{non-normalized canonical}.

The following lemma provides a lower bound for the gap between the cost of a non-canonical layout and a normalized canonical layout.
The idea is to consider a sequence of operations to make a non-canonical layout canonical (either normalized or non-normalized).
We analyze how large the gap is between the cost of the layout before the last operation of the sequence and the canonical cost.
\begin{lemma}
  \label{lemma:cost-gap}
  Let $\cost^*$ be the canonical cost of $I_\Phi$.
  Let $L$ be a non-canonical layout of $I_\Phi$.
  Then, $\cost(L) \geq \cost^* + 1 + (g - 1)^\alpha - g^\alpha$, where $g \geq 7$ is the minimum number of sources in a group.
\end{lemma}
\begin{proof}
  Let $L$ be a non-canonical layout of $I_\Phi$.
  Consider a sequence of operations as described in \ref{par:operations} that are applied to make $L$ normalized canonical.
  We first show that no operation in the sequence increases the cost.

  \begin{claim}
    No operation in the sequence increases the cost.
  \end{claim}
  \begin{claimproof}
    If there is a group $s$ with parent $v$ and $\dist(s, v) \geq 1$, then the cost for the edge from $s$ to $v$ is at least $1$.
    However, reconnecting $s$ to a new Steiner vertex $v_s$ that is connected to a sink $t$ with $\dist(s, t)$ costs $1$.
    Thus, the cost does not increase by applying Operation~1.

    For Operation~2, consider a Steiner vertex $v$ connected to sink $t$ such that there is a child $s$ of $v$ with $\dist(s, t) > 1$.
    If $v$ is only connected to sources of one group, then connecting $v$ to a sink $t^*$ with $\dist(s, t^*) = 1$ reduces the cost.
    Otherwise, $v$ is connected to sources of two source groups $A$ and $B$.
    We know that after Operation~1, $A$ and $B$ have distance $\sqrt{2}$ and there is a sink $t^*$ with $\dist(A, t^*) = 1$ and $\dist(B, t^*) = 1$.
    Consider the square that is uniquely defined by the corners $A$ and $B$; see Figure~\ref{fig:far-sink}.
    By Property~\ref{item:low-dist-to-intermediate}, it is $\dist(A,v) < 1$ and $\dist(B, v) < 1$ and thus $v$ must be contained within the square.
    But then, it is $\dist(v, t) \geq \dist(v, t^*)$ by construction, and applying Operation~2 (connecting $v$ to $t^*$) does not increase the cost.
%
      \begin{figure}
    \centering
    \includegraphics[page=3]{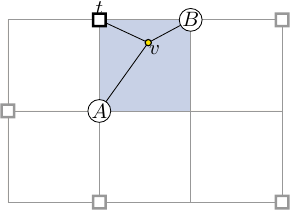}
    \caption{Two groups $A$ and $B$ with distance $\sqrt{2}$ that are a part of a literal, a corner or a clause gadget.
      The orange area indicates possible positions of their parent $v$.
      The sinks in gray mark the possible positions of other sinks that may be relevant.
      These are at least as far away from $v$ as $t$.}
    \label{fig:far-sink}
  \end{figure}

  For Operation~3, we have already shown in Lemma~\ref{lemma:no-split-op} that it does not increase the cost.

  Finally, for Operation~4, consider two source groups $A$ and $B$ with $|B| \geq |A| \geq g \geq 7$ that share a parent $v$.
  By Property~\ref{item:distance-to-sink}, $A$ and $B$ have distance $\sqrt{2}$ and are connected to a sink $t$ that has distance $1$ to both $A$ and $B$.
  By Corollary~\ref{cor:cost-difference}, the cost decrease is at least
  $\min(1, |A| \dist(A, B) + \dist(B, t) - 2, |A|(\dist(A, t) - 1) + |B|\dist(B,t - 1)) \geq \min(1, 7 \sqrt{2} + 1 - 2, 0) \geq 0$.
  Thus, the cost does not increase by Operation~4.
  \end{claimproof}

  It remains to show that the cost difference between $L$ and a normalized canonical layout is at least the proposed lower bound.
  For this, we consider the last layout before it is canonical (normalized or non-normalized).
  Denote this layout by $L'$.
  We compare the cost of $L'$ to the cost of a normalized canonical layout, and we distinguish the following cases depending on the operation that makes $L'$ canonical.
  Figure~\ref{fig:last-operations} shows for each possible last operation an example for $L'$ and $L'$ after the last operation.

  \begin{figure}
  \centering
  \includegraphics[page=2]{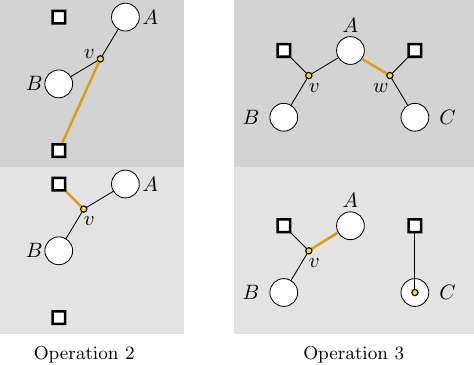}
  \caption{Examples of layouts before and after the last operation. Orange edges mark the change.}
  \label{fig:last-operations}
\end{figure}

  First, we note that Operation~1 is never the last operation as the resulting layout contains a Steiner vertex with a single child, and thus is not canonical.
  Operation~2 ensures that every source is connected to a sink with distance $1$ by changing the parent $t$ of a Steiner vertex $v$ to $t^*$.
  If we assume that such an operation is the last, then sources of the same group share a parent in $L'$.
  If the layout is normalized canonical afterwards, then in $L'$, $v$ is only connected to the sources of a single source group $A$, which has distance at least $\sqrt{5}$ by construction.
  The cost for the incident edges of $v$ in $L'$ is at least $|A|^\alpha \cdot \sqrt{5}$,
  while the cost for $v$ if connecting it to $t^*$ is $|A|^\alpha$.
  Thus, the cost decreases by at least $|A|^\alpha (\sqrt{5} - 1) \geq g^\alpha \geq 1$.

  If the layout is non-normalized canonical afterwards, then $v$ is connected to all sources of two source groups $A$ and $B$ with $\dist(A, B) = \sqrt{2}$ by Property~\ref{item:low-dist-to-intermediate}.
  We assume without loss of generality that $\dist(A, t) \geq \sqrt{5}$ and $\dist(B,t) \geq 1$.
  We compare the cost to the normalized canonical cost.
  By Corollary~\ref{cor:cost-difference}, the cost of a normalized canonical layout is less than $L'$ by at least
  \begin{align*}
    &\min(1, |A| \dist(A, B) + \dist(B, t) - 2,|A|(\dist(A, t) - 1) + |B|(\dist(B, t) - 1)) \\
    &\geq \min(1, 7 \cdot \sqrt{5} + 1 - 2, 7 \cdot (\sqrt{5} - 1)) \geq 1.
  \end{align*}
  Thus, in all cases, the cost of a normalized canonical layout is less than $L'$ by at least $1$.

  After Operation~3, sources of the same group share a parent.
  Consider the last such operation, and let $A$ be the corresponding group with subgroups $A_v$ (connected to Steiner vertex $v$) and $A_w$ (connected to Steiner vertex $w$).
  Since this operation makes the layout canonical, each group except for $A$ already shares a parent, and sources in $A$ are either connected to $v$ or $w$.
  Moreover, only source groups with distance $\sqrt{2}$ may share a parent, and each source has distance $1$ to its assigned sink.

  We first show that the cost does not increase if we assume that $v$ and $w$ are only connected to sources of $A$.
  First, assume that $v$ is connected to two different source groups $A_v$ and $B$ with $|A_v| \geq 1$ and $|B| \geq g$.
  By Corollary~\ref{cor:cost-difference}, the cost does not increase if we connect $A_v$ and $B$ using distinct Steiner vertices since the cost difference is at least
  \begin{align*}
    &\min(1, |A_v| \dist(A_v, B) + \dist(B, t) - 2, |A_v|(\dist(A_v, t) - 1) + |B|(\dist(B, t) - 1)) \\
    &\geq \min(1, \sqrt{2} + 1 - 2, 0) \geq 0.
  \end{align*}
  The same holds for $w$ if $w$ is connected to sources of $A_w$ and $B$.
  Thus, we can assume that $v$ and $w$ are only connected to $A_v$ and $A_w$, respectively.
  In this case, the total cost in $L'$ for both is $|A_w|^\alpha + |A_v|^\alpha$ while after the operation, the cost for connecting $A$ is $|A|^\alpha$.
  The cost difference is minimized for $|A_v| = 1$, which gives us a cost decrease of at least
  \begin{equation*}
    (|A| - 1)^\alpha + 1 - |A|^\alpha \geq (g - 1)^\alpha + 1 - g^\alpha .
  \end{equation*}
  Note that Operation~4 is never needed to make a non-canonical layout non-normalized canonical.
  In total, we have shown that a non-canonical layout costs at least $\min(1 + (g - 1)^\alpha - g^\alpha, 1) = 1 + (g - 1)^\alpha - g^\alpha$ more than a normalized canonical layout, where $g$ is the minimum number of sources in a group.
\end{proof}
To summarize, it follows from Lemma~\ref{lemma:cost-difference} that any layout of the \flamecast instance $I_\Phi$ that has at most the canonical cost is canonical.
Thus, $I_\Phi$ has a valid layout of canonical cost if and only if it has a valid canonical layout.  Due to Lemma~\ref{lemma:np-hard-correctness}, this is the case if and only if $\Phi$ is satisfiable.
This immediately yields the following theorem.
\begin{theorem}
  \label{thm:np-hard-three-layers}
  \flamecast is NP-hard for one intermediate layer and $\alpha < 1$.
\end{theorem}
By a simple reduction, we generalize this result for more layers.
\begin{corollary}
  \flamecast is NP-hard for instances with at least one intermediate layer and $\alpha < 1$.
\end{corollary}
\begin{proof}
  We show that \flamecast with $\lambda$ intermediate layers can be polynomially reduced to \flamecast with $\lambda + 1$ intermediate layers.
  Let $I$ be an instance of \flamecast with $\lambda \in \mathbb{N}_0$ intermediate layers for some $\alpha \in [0, 1)$.
  We construct an instance $I'$ of \flamecast with one more layer by adding a layer after the sources with capacity $1$.
  We show that if $I'$ is feasible, then there is a cost-optimal layout for $I'$ where the underlying forest has $\lambda$ intermediate layers.
  Let $L$ be a valid layout of $I'$ with $\lambda + 1$ intermediate layers.
  By construction, every Steiner vertex in layer $\lambda + 1$ has a single source as child.
  Removing such a Steiner vertex $v$ and its incident edges and directly connecting the child of $v$ to the parent of $v$ does not increase the cost.
  By repeating this for every Steiner vertex in layer $\lambda + 1$, we obtain a valid layout with $\lambda$ intermediate layers whose cost does not exceed the cost of $L$.
  As \flamecast with $\lambda = 1$ is NP-hard by Theorem~\ref{thm:np-hard-three-layers}, this extends to any $\lambda \geq 1$.
\end{proof}
The gap shown in Lemma~\ref{lemma:cost-difference} additionally gives us a lower bound for an approximation factor if P $\neq$ NP.
\begin{theorem}
  \label{thm:approx-np-hard}
  Approximating \flamecast within a factor of $1 + \frac{1 + (g - 1)^\alpha - g^\alpha}{n \cdot (2g)^\alpha}$ is NP-hard, where $n$ is the number of sources and $g \geq 7$ is the minimum size of groups.
\end{theorem}
\begin{proof}
  Let $n_g +n_{2g}$ be the number of source groups.
  By Lemma~\ref{lemma:canonical-cost}, the optimal cost for $I_\Phi$ if $\Phi$ is a yes-instance is at most $\opt = n_g \cdot g^\alpha + n_{2g} \cdot (2g)^\alpha < n \cdot (2g)^\alpha$.
  On the other hand, if $\Phi$ is a no-instance, then the optimal cost for $I_\Phi$ is at least $\opt + 1 + (g - 1)^\alpha - g^\alpha$.
  This means that if there is a polynomial algorithm with an approximation factor of less than
  \begin{align*}
    \frac{\opt + 1 + (g - 1)^\alpha - g^\alpha}{\opt} = 1 + \frac{1 + (g - 1)^\alpha - g^\alpha}{\opt} > 1 + \frac{1 + (g - 1)^\alpha - g^\alpha}{n \cdot (2g)^\alpha} ,
  \end{align*}
  then \textsc{Planar monotone 3-SAT} can be solved in polynomial time.
\end{proof}
By choosing $g = 7$ as the minimum number of sources in a group, \flamecast cannot be approximated within a factor of $1 + \frac{1 + 6^\alpha - 7^\alpha}{14^\alpha n}$ for $\alpha \in [0, 1)$, unless $\text{P}=\text{NP}$.
For $\alpha = 0$, this is $1 + \frac{1}{n}$, and it decreases for larger $\alpha$.

\section{Convex Instances}
\label{sec:convex}

In this section, we consider instances where the sources are in convex position.
More formally, an instance $I = (S, T, C, \alpha)$ is \emph{convex} if all sources in $S$ are located on the boundary of the convex hull of $S \cup T$.
We show that \flamecast remains NP-hard in such a setting.
In fact, we show that \flamecast is hard even if all sources lie on a circle and all sinks on the circle's center.
We call this class of instances \emph{circular instances}.
If additionally, the sources are distributed evenly on the circle, then we call the instance a \emph{source-equally-spaced} circular instance.
If the sources are divided into groups, where the sources of the same group share a position, and the groups are distributed evenly on the circle, then we call the instance a \emph{group-equally-spaced} circular instance.
Note that every source-equally-spaced circular instance is also a group-equally-spaced instance but not vice versa.

We show two NP-hardness results for different types of $\alpha$ for group-equally-spaced circular instances with one intermediate layer.
The first result shows NP-hardness if $\alpha < 1$ converges sufficiently fast to $1$ for $n \to \infty$, while the second result proves NP-hardness for constant $\alpha$.
Recall from Theorem~\ref{thm:poly-case} that \flamecast can be solved in polynomial time if $\alpha = 1$.
Thus, our hardness proofs only leave a small gap for the case where $\alpha$ goes slowly to $1$ for $n \to \infty$.
Moreover, we prove that there is no polynomial approximation algorithm for group-equally-spaced circular instances with one intermediate layer with approximation factor $1 + \frac{1}{n^2}$ unless $\text{P} = \text{NP}$.

We complement these negative results with two positive results.
First, we propose a polynomial approximation algorithm with approximation factor $1 + \frac{1}{n^2}$ and runtime $O(n^2\log^3(n))$ for source-equally-spaced circular instances with one intermediate layer and one sink.
The quality of the approximation only depends on the quality of the chosen algorithm to embed a given topology, i.e., if the embedding algorithm is optimal, then we find a cost-optimal layout.
Some observations leading to this result can be generalized to convex instances with some caveats.
For this setting with one intermediate layer, one sink, no capacities and $\alpha = 0$, we propose a polynomial approximation algorithm that approximates a solution within a factor of $1 + \varepsilon$ in $O(n^8\log^3(n))$ time.
As for the circular case, the quality of the approximation only depends on the quality of the embedding.

\subsection{Hardness of Group-Equally-Spaced Circular Instances with One Intermediate Layer}
For both constant $\alpha$ and sufficiently large $\alpha$, we reduce from \textsc{3-Partition}, a well-known strongly NP-hard problem~\cite{garey2002}.

\subparagraph{3-Partition}
The input consists of an integer $t$ and a (multi-)set $Z=\set{z_1,\dots, z_{m}}$ of $m=3k$ positive integers with $\frac{t}{4} < z_i < \frac{t}{2}$ for all $i \in [3k]$ such that $\sum_{i=1}^{m} z_i = k \cdot t$.
We ask if there is a partition of $Z$ into triplets $S_1,\dots,S_k$ such that each triple sums to $t$, i.e., $\sum_{z \in S_i} z = t$ for all $i \in k$.

\subparagraph{Reduction}
Given an instance $Z = \set{z_1, \ldots, z_m}$ of \textsc{3-Partition} with $m = 3k$ for some $k\in \mathbb{N}$, we construct a circular instance $I_Z$ of \flamecast with a similar construction for both hardness results.
The instance consists of $k$ sinks and $m$ groups of sources where group $i \in \set{1, \ldots, m}$ consists of $z_i + \hat{c}$ sources for a globally chosen offset $\hat{c}$.
Sources of the same group share the same position, and the groups are distributed evenly on the unit circle.
The sinks have capacity $t + 3 \cdot \hat{c}$ each, and the intermediate layer has capacity $\frac{t}{2} + \hat{c}$.
We ask if there is a valid layout that costs at most $\sum_{i = 1}^{m} (z_i + \hat{c})^\alpha$.

We call a layout \emph{canonical} if the sources of each source group are connected to a Steiner vertex that is located at the group.
The cost of such a layout is $\sum_{i = 1}^m (z_i + \hat{c})^\alpha$, and thus, any valid canonical layout is a solution to the constructed instance $I_F$.
We call $\sum_{i = 1}^m (z_i + \hat{c})^\alpha$ the \emph{canonical cost}.
Figure~\ref{fig:np-circle} shows such an instance and a corresponding canonical layout.
  \begin{figure}
    \centering
    \includegraphics[page=3]{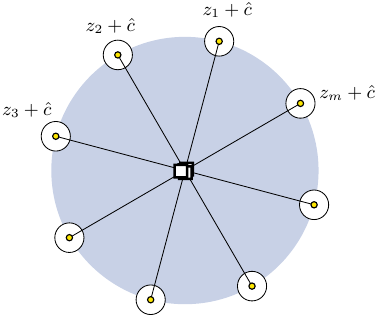}
    \caption{A canonical layout of a group-equally-spaced circular instance, where each circle represents a group with $z_i + \hat{c}$ sources. The sinks are located in the center. }
    \label{fig:np-circle}
  \end{figure}

We later show for both hardness results separately that any layout with cost at most the canonical cost is canonical.
Assuming that this is true for now, we argue why the reduction is correct.
In a canonical layout, there are $m$ intermediate vertices with loads $z_1 +\hat{c}, \ldots, z_m + \hat{c}$.
Note that the capacity constraints in the intermediate layer are not violated.
The only choices we make are which intermediate vertices to connect to which sink such that the sink capacities are not exceeded.
Since there are $k$ sinks with capacity $t + 3 \cdot \hat{c}$ each, assigning intermediate vertices to sinks is equivalent to partitioning the elements $\set{z_1 +\hat{c}, \ldots, z_m + \hat{c}}$ into $k$ triples with sum $t + 3 \cdot \hat{c}$ each.
This is equivalent to partitioning $Z$ into $k$ triples with sum $t$ each, which directly yields the following lemma.
\begin{lemma}
  \label{lemma:np-circle-correctness-large-alpha}
  The \textsc{3-Partition} instance $Z$ has a solution if and only if there is a valid canonical layout for the \flamecast instance $I_Z$.
\end{lemma}
It remains to show that the cost of every non-canonical layout is higher than the canonical cost, which we show for constant $\alpha$ and sufficiently large $\alpha$ separately.

\subsubsection{Sufficiently Large \texorpdfstring{\boldmath$\alpha$}{𝛼}}
In this setting, we choose $\hat{c} = 0$ and a sufficiently large $\alpha$, namely $\alpha \geq \log_n\left(m - 1 + \cos\left(\frac{\pi}{m}\right)\right)$, depending on the number of source groups $m$ and number of sources $n$.
We later show that $\alpha$ asymptotically behaves like $1 - \frac{1}{O(n^3\log(n))}$ for $n \to \infty$.
The idea to prove that any optimal layout is canonical is as follows.
In a non-canonical layout, there are sources of different groups that share one Steiner vertex.
These sources take a detour on the way to their sinks compared to a canonical layout, whereas the path from each source to the sink is a straight line.
On the other hand, the total cost may decrease if an edge to the sink is shared by more sources.
We however show that the benefits of sharing are not worth the detour.
\begin{lemma}
  \label{lemma:np-circle-canonical-large-alpha}
  Let $I$ be a group-equally-spaced circular instance of \flamecast with one intermediate layer and $n$ sources that are split into $m$ groups.
  If $\alpha \geq \log_n\left(m - 1 + \cos\left(\frac{\pi}{m}\right)\right)$, then the cost of every non-canonical layout is more than the canonical cost.
\end{lemma}
\begin{proof}
  Consider a layout $L$ for $I$ that is not canonical and denote the sizes of the source groups by $z_1, \ldots, z_m$.
  We prove that the cost of $L$ is larger than the canonical cost.
  In the following, we may assume that sources of the same group share a parent in $L$ since otherwise there is a layout $L'$ with strictly lower cost by Lemma~\ref{cor:no-split}.
  If every group in $L$ is connected to a distinct Steiner vertex, it is clear that a canonical layout is optimal if we additionally place the Steiner vertices at each group's position.

  Otherwise, there are Steiner vertices in $L$ that are connected to sources of at least two different groups.
  Let $r$ be the maximum distance between such a Steiner vertex and a sink.
  If $r = 0$, then every Steiner vertex is placed at the center, and the cost of $L$ is $\sum_{i=1}^{m} z_i$ which is larger than $\sum_{i=1}^{m}z_i^\alpha$ for $\alpha < 1$.
  Thus, we assume $r > 0$ in the following.

  Sources that are connected to Steiner vertices with children from multiple groups take a detour compared to how they are connected in a canonical layout.
  On the other hand, the cost for connecting sources is reduced if the sources are bundled by a Steiner vertex.
  In the following, we first give an upper bound for how much sharing benefits us and then give a lower bound for how much we pay for the detours.

  We denote the cost for $L$ by $\cost_L(\alpha)$, depending on $\alpha$.
  We quantify how much we save at most by bundling different groups with one Steiner vertex, which is bounded by how much we save by enabling sharing (setting $\alpha < 1$ instead of $\alpha = 1$).
  This is given by the cost difference $\cost_L(1) - \cost_L(\alpha)$.
  Let $m'$ be the number of intermediate vertices $v_1, \ldots, v_{m'}$ in $L$.
  Moreover, we denote the distance between the origin and Steiner vertex $v_i$ by $r_i$ and the number of sources that are connected to $v_i$ by $b_i$ for each $i \in \set{1, \ldots, m'}$.
  Compared to the cost of $L$ with $\alpha = 1$, the cost of each edge between an Steiner vertex $v_i$ and its sink decreases by $r_i \cdot \left(b_i - b_i^\alpha\right)$.
  In total, the cost reduces by $\sum_{i = 1}^{m'} r_i \cdot \left(b_i - b_i^\alpha\right)$ if we have $\alpha < 1$.
  Thus, we get the following inequality.
  \begin{align*}
    \cost_L(1) - \cost_L(\alpha) &=  \sum_{i = 1}^{m'} r_i \cdot \left(b_i - b_i^\alpha\right) \\
    &\leq r \sum_{i = 1}^{m'} \left(b_i - b_i^\alpha\right) \\
    &< r \left(\sum_{i = 1}^{m'} b_i - \left(\sum_{i = 1}^{m'} b_i\right)^\alpha\right) \\
    &= r \left(n - n ^\alpha\right) ,
  \end{align*}
  which yields the inequality
  \begin{align}
    \label{eq:compare-caL-c1L}
    \cost_L(\alpha) > \cost_L(1) - r \left(n - n ^\alpha\right) .
  \end{align}

  In the following, we determine the lower bound of the detour in $L$ compared to a canonical layout with $\alpha = 1$, which is given by $\cost_L(1) - \sum_{i = 1}^{m} z_i$.
  To bound this cost difference, we only consider the detour taken by a single source.
  Consider a Steiner vertex~$v$ with maximum distance $r$ from the center that is connected to sources of at least two different groups.

  Let $v$ be connected to two sources $a$ and $b$ that are not in the same group.
  We assume without loss of generality that $\dist(v, a) \geq \dist(v, b)$.
  Furthermore, let $d$ be the distance between Steiner vertex $v$ and source $a$.
  This can be seen in Figure~\ref{fig:np-circle-large-alpha}.
  \begin{figure}
    \centering
    \includegraphics{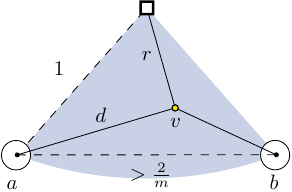}
    \caption{Two sources of different groups connected to the same parent~$v$.}
    \label{fig:np-circle-large-alpha}
  \end{figure}

  The path from $a$ to its sink in $L$ is longer by $d + r - 1$ compared to the length of the path in a canonical layout.
  Thus, we have
  \begin{align}
    \label{eq:compare-c1L-ca}
    \cost_L(1) \geq \sum_{i=1}^{m}z_i + d + r - 1 > \sum_{i=1}^{m}z_i^\alpha + d + r - 1
  \end{align}
  for $\alpha < 1$.
  Putting the inequalities of \ref{eq:compare-caL-c1L} and \ref{eq:compare-c1L-ca} together, we get
  \begin{align*}
    \cost_L(\alpha) &> \cost_L(1) - r\left(n - n^\alpha\right) > \sum_{i=1}^{m}z_i^\alpha + d + r - 1 - r\left(n - n^\alpha\right)
  \end{align*}
  It remains to show that the detour costs more than what sharing saves, i.e.
    $d + r - 1 \geq r \cdot \left(n - n^\alpha\right)$
  for our chosen $\alpha$.
  Since the angle between $a$ and $v$ at the center is at least~$\frac{\pi}{m}$, we get
    $d > \sqrt{1 + r^2 - 2r\cos\left(\frac{\pi}{m}\right)}$
  using the law of cosines.
  With this and our chosen $\alpha \geq \log_n(m - 1 + \cos(\frac{\pi}{m}))$, the rest of the proof is just substituting and rearranging as follows
  \begin{align*}
    d + r - 1 - r \cdot \left(n - n^\alpha\right) &> \sqrt{1 + r^2 - 2r\cos\left(\frac{\pi}{m}\right)} + r - 1 - r \cdot \left(m - \left(m - 1 + \cos\left(\frac{\pi}{m}\right)\right)\right) \\
    &= \sqrt{1 + r^2 - 2r\cos\left(\frac{\pi}{m}\right)} - 1 + r\cos\left(\frac{\pi}{m}\right) \\
    &> \sqrt{1 + r^2 \cos\left(\frac{\pi}{m}\right)^2 - 2r\cos\left(\frac{\pi}{m}\right)} - 1 + r\cos\left(\frac{\pi}{m}\right) \\
    &= 1 - r \cos\left(\frac{\pi}{m}\right) - 1 + r\cos\left(\frac{\pi}{m}\right) \\
    &= 0 .
  \end{align*}
  This shows that even the best possible sharing does not account for the detour that the single source $a$ takes.
  Thus, every non-canonical layout costs more than a canonical layout.
\end{proof}
Together with Lemma~\ref{lemma:np-circle-correctness-large-alpha}, this proves the correctness of the reduction.
Since the reduction is polynomial, we obtain the following theorem.

\begin{theorem}
  \flamecast is NP-hard on group-equally-spaced circular instances with one intermediate layer and $\alpha \geq \log_n(m - 1 + \cos\left(\frac{\pi}{m}\right))$ where $n$ is the number of sources.
\end{theorem}
The following observation describes the behavior of $\alpha$ as a polynomial function in $n$.
\begin{observation}
  \label{obs:alpha-asymptotic}
  The function $f(n) = \log_n(n - 1 + \cos\left(\frac{\pi}{n}\right))$ asymptotically behaves like $1 - \frac{1}{O(n^3 \log(n))}$ for $n \to \infty$.
\end{observation}
\begin{proof}
  We first rewrite $f(n)$ to
  \begin{equation*}
    1 - \left(1 - \frac{\log(n - 1 + \cos\left(\frac{\pi}{n}\right))}{\log(n)}\right) = 1 - \frac{\log(n) - \log(n - 1 + \cos\left(\frac{\pi}{n}\right))}{\log(n)} = 1 - \frac{\log\left(\frac{n}{n - 1 + \cos\left(\frac{\pi}{n}\right)}\right)}{\log(n)}
  \end{equation*}
  and show for $g(n) = n^3\log(n)$ and $h(n) = \frac{\log(n)}{\log\left(\frac{n}{n - 1 + \cos\left(\frac{\pi}{n}\right)}\right)}$ that
  \begin{equation*}
    \lim_{n \to \infty} \frac{g(n)}{h(n)} = \lim_{n \to \infty} n^3 \cdot \log\left(\frac{n}{n - 1 + \cos\left(\frac{\pi}{n}\right)}\right)
  \end{equation*}
  is constant.
  For this, we use the series expansion of $\cos\left(\frac{\pi}{n}\right)$ at $\infty$, which is
  \begin{equation*}
    \sum _{k=0}^{\infty } \frac{(-1)^k \pi ^{2 k} \left(\frac{1}{n}\right)^{2 k}}{(2 k)!} .
  \end{equation*}
  This provides a lower bound if we take the first two terms
  \begin{equation*}
    \cos\left(\frac{\pi}{n}\right) \geq 1 - \frac{\pi^2}{2n^2}
  \end{equation*}
  and an upper bound if we take the first three terms
  \begin{equation*}
    \cos\left(\frac{\pi}{n}\right) \leq 1 - \frac{\pi^2}{2n^2} + \frac{\pi^4}{24n^4} .
  \end{equation*}
  Using the lower bound, we get
  \begin{align*}
    \lim_{n \to \infty} \frac{g(n)}{h(n)} &= \lim_{n \to \infty} n^3 \cdot \log\left(\frac{n}{n - 1 + \cos\left(\frac{\pi}{n}\right)}\right)
    \leq \lim_{n \to \infty} n^3 \cdot \log\left(\frac{n}{n - 1 + 1 - \frac{\pi^2}{2n^2}}\right) \\
    &= \lim_{n \to \infty} n^3 \cdot \log\left(\frac{n}{n - \frac{\pi^2}{2n^2}}\right)
    = \lim_{n \to \infty} \frac{\log\left(\frac{n}{n - \frac{\pi^2}{2n^2}}\right)}{\frac{1}{n^3}} .
    \intertext{Since both the numerator and the denominator go to $0$ for $n \to \infty$ and both are differentiable for $n > 0$, we apply L'Hospital's rule to obtain}
    &= \lim_{n \to \infty} \frac{\frac{3\pi^2}{\pi^2n - 2n^4}}{\frac{-3}{n^4}} = \lim_{n \to \infty} \frac{3\pi^2n^4}{6n^4-3\pi^2n} = \frac{\pi^2}{2} .
  \end{align*}
  We use the upper bound analogously and get
    \begin{align*}
    \lim_{n \to \infty} \frac{g(n)}{h(n)} &= \lim_{n \to \infty} n^3 \cdot \log\left(\frac{n}{n - 1 + \cos\left(\frac{\pi}{n}\right)}\right) \\
    &= \lim_{n \to \infty} n^3 \cdot \log\left(\frac{n}{n - \frac{\pi^2}{2n^2} + \frac{\pi^4}{24n^4}}\right)
    = \lim_{n \to \infty} \frac{\log\left(\frac{n}{n - \frac{\pi^2}{2n^2} + \frac{\pi^4}{24n^4}}\right)}{\frac{1}{n^3}} .
    \intertext{Again, we apply L'Hospital's rule to obtain}
    &= \lim_{n \to \infty} \frac{\frac{5 \pi ^4-36 \pi ^2 n^2}{24 n^6-12 \pi ^2 n^3+\pi ^4 n}}{\frac{-3}{n^4}} = \lim_{n \to \infty} \frac{36\pi^2n^6-5\pi^4n^4}{72n^6 - 36\pi^2n^3 + 3\pi^4n} = \frac{\pi^2}{2} .
  \end{align*}
  As the bounds match, we have
  \begin{equation*}
    \lim_{n \to \infty} \frac{g(n)}{h(n)}= \frac{\pi^2}{2} ,
  \end{equation*}
  which is a constant.
\end{proof}

\subsubsection{Constant \texorpdfstring{\boldmath$\alpha$}{𝛼}}
\label{sec:hardn-circ-inst-const-alpha}

Let $\alpha$ be fixed but arbitrary in the following.
The offset $\hat{c}$ is suitably chosen for every instance $I_Z$, depending on the \textsc{3-Partition} instance $Z$.
In the following, we show that every optimal layout is canonical if we choose $\hat{c}$ large enough.
For this, we define operations that make a non-canonical layout canonical and prove that each operation strictly reduces the cost.
\begin{description}
    \item[Operation 1] Let $A$ be a source group, where a subset $A_v \subseteq A$ is connected to a Steiner vertex $v$ and a subset $A_w \subseteq A$ is connected to a Steiner vertex $w \neq v$.
    Either reconnect each source in $A_v$ to $w$ or vice versa.
    \item[Operation 2] Let $v$ be a Steiner vertex connected to two distinct source groups $A$ and $B$.
    Add a Steiner vertex $v_A$ at the location of $A$, connect $v_A$ to a sink and reconnect $A$ to $v_A$.
    Repeat for $B$, adding a Steiner vertex $v_B$ at the location of $B$.
  \end{description}
We apply each type of operation exhaustively before moving to the next one.
It is clear that the resulting layout is canonical.
The following lemma shows how much the cost difference between a canonical and a non-canonical layout is.
%
\begin{lemma}
  \label{lemma:circular-cost-diff}
  Let $\hat{c} \geq \max\left(\sqrt[1-\alpha]{2m}, \frac{t}{2}\right)$.
  The cost of any non-canonical layout for $I_Z$ is more than $\frac{1}{m}$ higher than the canonical cost.
\end{lemma}
\begin{proof}
  Let $L$ be a non-canonical layout for $I_Z$.
  Consider the sequence of operations we apply to $L$ as specified.
  We distinguish two cases depending on the type of the last operation before the layout is canonical.
  We denote the layout before the last operation by $L'$.

  We begin with Operation~2 and show that applying Operation~2 reduces the cost by at least $2$.
  If Operation~2 is applied, each group already shares the same parent.
  Let $A$ and $B$ be the groups that are split by the last operation.
  Denote the shared parent of $A$ and $B$ in $L'$ by $v$.
  Then, connecting the sources in $A$ and $B$ to $v$ in $L'$ costs at least
  \begin{align*}
    |A| \cdot \dist(A, v) + |B| \cdot \dist(B, v) &\geq \hat{c} \cdot \dist(A, v) + \hat{c} \cdot\dist(B, v)
    \geq \hat{c} \cdot \dist(A, B) ,
  \end{align*}
  while connecting $A$ and $B$ in $L$ using two distinct Steiner vertices costs $|A|^\alpha + |B|^\alpha$.
  Using basic geometry, the distance between two groups of sources is at least $2\sin(\frac{\pi}{m}) > \frac{2}{m}$.
  Moreover, we know that $|A|, |B| \leq \hat{c} + \frac{t}{2}$ by construction.
  With $\hat{c} \geq \frac{t}{2}$, it is $|A|^\alpha + |B|^\alpha \leq 2 \left(\hat{c} + \frac{t}{2}\right)^\alpha \leq 2 \cdot (2\hat{c})^\alpha$.

  Thus, the cost difference is at least
\begin{align*}
  \hat{c} \cdot \dist(A, B) - |A|^\alpha - |B|^\alpha
  &\geq \hat{c} \cdot \frac{2}{m} - 2 \cdot (2\hat{c})^\alpha \\
  &= \hat{c}^\alpha \cdot \hat{c}^{1 - \alpha} \cdot \frac{2}{m} - 2 \cdot (2\hat{c})^\alpha \\
  &\geq \hat{c}^\alpha \cdot (2m)^{\frac{1 - \alpha}{1-\alpha}} \cdot \frac{2}{m} - 2 \cdot (2\hat{c})^\alpha \\
  &= 4\hat{c}^\alpha - 2 \cdot (2 \hat{c})^\alpha \\
  &\geq \left(4 - 2 \cdot 2^\alpha\right) (2m)^{\frac{\alpha}{1-\alpha}} ,
\end{align*}
which is at least $2$ for $m \geq 1$.
This means that applying Operation~2 reduces the cost by at least $2$.

By Lemma~\ref{lemma:no-split-op}, the cost does not increase by applying Operation~1.
  If Operation~1 is the last operation, then there is a group $A$ in $L'$ that is connected to exactly two distinct Steiner vertices $v$ and $w$.
  Since we assume that only one more operation is needed, it is not possible that both $v$ and $w$ are connected to sources of another group.
  However, it is possible that either $v$ or $w$ is connected to sources of another group in addition to $A$.

  We start with the easier case, where both $v$ and $w$ are only connected to sources in $A$.
  Then, the cost of connecting $A$ in $L'$ is $|A_v|^\alpha + |A_w|^\alpha$, while it is $|A|^\alpha$ after the last operation.
  Thus, the cost difference is $|A_v|^\alpha + |A_w|^\alpha - |A|^\alpha$, which is minimal for $|A_v| = |A| - 1$ and $|A_w| = 1$.
  This means that the cost difference in this case is at least $1 - (|A|^\alpha - (|A| - 1)^\alpha)$.

  In the other case, we assume that $w$ is shared with another group $B$ in $L'$.
  Since we only consider the last operation, we know that in $L'$, all sources in $B$ are connected to $w$, that $v$ is only connected to sources in $A$, and that all sources in $A_w$ are reconnected to $v$.
  This is shown in Figure~\ref{fig:np-circle-constant-alpha-op1}.
    \begin{figure}
    \centering
    \includegraphics[page=4]{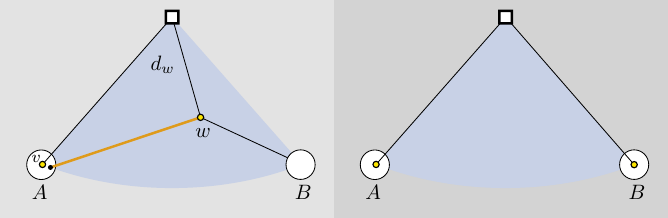}
    \caption{The layout $L'$ before and after the last operation in the sequence, Operation~1.}
    \label{fig:np-circle-constant-alpha-op1}
  \end{figure}
  By Lemma~\ref{lem:properties:majority-weber-point}, in an optimal embedding, $v$ is located at the same position as $A$.
  Let $d_w$ be the distance of $w$ to its sink.
  Then, the total cost for connecting $A$ and $B$ in $L'$ consists of the cost for connecting $A_v$ via $v$ and the cost for connecting $A_w$ and $B$ via $w$, which is
  \begin{align*}
    \cost(L') = |A_v|^\alpha + |A_w| \dist(A, w) + |B| \dist(B, w)  + d_w (|B| + |A_w|)^\alpha.
  \end{align*}
  This is minimal for $|A_w| = 1$ since all sources in $A_w$ are reconnected to $v$ in the last operation.
  Thus, the cost difference between $L'$ and a canonical layout in this case is at least
  \begin{align}
    \label{eq:cost-difference}
    (|A| - 1)^\alpha + \dist(A, w) + |B| \dist(B, w)  + d_w (|B| + 1)^\alpha - |A|^\alpha - |B|^\alpha .
  \end{align}
  Note that if $|B| = 0$ and if $w$ is located at $A$, this covers exactly the case where $w$ is only connected to sources of group $A$.
  Thus, it is sufficient to determine the cost difference in this more general case.
  We first show that for $|A_w| = 1$, $w$ is located at the same position as group $B$, i.e., $\dist(B, w) = 0$ and $d_w = 1$.
  It is
  \begin{align*}
    \sum_{\substack{v \in N(w) \\ \pos(v) = \pos(B)}} \ell(vw)^\alpha = |B| \qquad \text{and} \qquad \sum_{\substack{v \in N(w) \\ \pos(v) \neq \pos(B)}} \ell(vw)^\alpha = 1 + (|B| + 1)^\alpha .
  \end{align*}
  Since
  \begin{align*}
    |B| = |B|^\alpha \cdot |B|^{1- \alpha} \geq |B|^\alpha \cdot \hat{c}^{1-\alpha} \geq |B|^\alpha \cdot 2m \geq (|B| + 1)^\alpha + 1
  \end{align*}
  holds for all $m \geq 1$, the Steiner vertex $w$ shares the same position as $B$ by Lemma~\ref{lem:properties:majority-weber-point}.

  Plugging $d_w = 1$, $\dist(B, w) = 0$ and $\dist(A, w) = \dist(A,B) > \frac{2}{m}$ into Equation~\ref{eq:cost-difference}, the cost difference between $L'$ and the canonical cost is more than
  \begin{align*}
    &(|A| - 1)^\alpha + \frac{2}{m} + (|B| + 1)^\alpha - |A|^\alpha - |B|^\alpha\\
    &= \frac{2}{m} - (|A|^\alpha - (|A| - 1)^\alpha) + (|B| + 1)^\alpha - |B|^\alpha \\
    &\geq \frac{2}{m} - (|A|^\alpha - (|A| - 1)^\alpha) .
  \end{align*}

  We prove $|A|^\alpha - (|A| - 1)^\alpha < \frac{1}{m}$.
  Consider the function $f(x) = x^\alpha$ for $\alpha \in [0, 1)$.
  The derivative of $f(x)$ is $f'(x) = \alpha x^{\alpha - 1}$.
  By the mean value theorem, there is a $z \in (x, x+1)$ such that $f'(z) = \frac{f(x + 1) - f(x)}{x + 1 - x} = f(x + 1) - f(x)$.
  Let $z \in (|A| - 1, |A|)$ such that $f'(z) = f(|A|) - f(|A| - 1)$.
  Since $f'$ is strictly decreasing, it is
  \begin{align*}
    |A|^\alpha - (|A| - 1)^\alpha = f'(z) \leq f'(|A| - 1) \leq f'(\hat{c}) = \alpha \cdot (2m)^{\frac{\alpha - 1}{1 - \alpha}} = \frac{\alpha}{2m} < \frac{1}{m} .
  \end{align*}
  Thus, if the last operation is Operation~1, then the cost is reduced by at least $\frac{2}{m} - (|A|^\alpha - (|A| - 1)^\alpha) > \frac{2}{m} - \frac{1}{m} = \frac{1}{m}$.

  In total, the cost difference between a non-canonical layout and a canonical layout is more than $\frac{1}{m}$.
\end{proof}
It follows from the previous lemma that any cost-optimal layout for $I_Z$ is canonical, and thus, the reduction is correct if we choose a sufficiently large value for $\hat{c}$.
For constant $\alpha$, the reduction is polynomial since the size of $\hat{c}$ is polynomial in the input.
Combining these findings, we obtain the following theorem.
\begin{theorem}
  The \flamecast problem for fixed but arbitrary $\alpha < 1$ is NP-hard, even if restricted to group-equally-spaced circular instances with one intermediate layer.
\end{theorem}
Since $\hat{c}$ grows exponentially if $\alpha$ increases with the input size, we cannot easily generalize this approach to non-constant $\alpha$.

Additionally, the previous lemma also implies that every valid layout for $I_Z$ that corresponds to a no-instance of \textsc{3-Partition} is at least $\frac{1}{m}$ more expensive than the canonical cost.
This gives us a lower bound on an approximation factor of a polynomial approximation algorithm, assuming $\text{P}\neq\text{NP}$.
\begin{theorem}
  \label{thm:circular-approx-hard}
  If $\text{P}\neq\text{NP}$, then there is no polynomial approximation algorithm for group-equally-spaced circular instances with one intermediate layer with approximation factor $1 + \frac{1}{n^2}$, where $n$ is the number of sources.
\end{theorem}
\begin{proof}
  Let $\opt$ be the canonical cost.
  By \cref{lemma:circular-cost-diff}, if $Z$ is a no-instance, then every valid layout for $I_Z$ costs more than $\opt + \frac{1}{m}$, where $m$ is the number of source groups.
  Since $\opt \leq n$ and $m < n$, we get
   \begin{align*}
    \frac{\opt + \frac{1}{m}}{\opt} = 1 + \frac{1}{\opt \cdot m} > 1 + \frac{1}{n^2} .
  \end{align*}

  This means that if there is a polynomial approximation algorithm with approximation factor $1 + \frac{1}{n^2}$, then 3-\textsc{Partition} is solvable in polynomial time.
\end{proof}

\subsubsection{Discussion}
We have shown for both sufficiently large $\alpha$ that increases depending on the number of sources and for constant $\alpha$ that \flamecast on circular instances is NP-hard.
Both reductions boil down to the fact that any optimal layout is canonical for the constructed instances, i.e., that sources of distinct groups are connected to distinct Steiner vertices.
Although the construction for both types of $\alpha$ is similar, the reasoning why any optimal layout is canonical is very different.
Very roughly speaking, sharing an edge for sufficiently large $\alpha$ is not helpful enough that it is worth a detour.
On the other hand, for constant $\alpha$ and sufficiently large groups of sources, sharing is so helpful that we want to share edges as soon as possible, i.e., immediately at a group.
This difference in arguments explains the gap that exists for instances with $\alpha$ that grows depending on $n$ but not sufficiently fast.
In this gap, it may be beneficial to accept a small detour to share an edge with more sources.

\subsection{Solving Circular and Convex Instances}
So far we have seen that \flamecast is NP-hard to approximate within a certain factor even when restricted to circular instances.
This leaves the question of whether there are non-trivial families of instances that can be solved efficiently.

In this section, we propose dynamic programs for source-equally-spaced circular instances with one intermediate layer and one sink, as well as for convex instances with one intermediate layer, one sink, unlimited capacity and $\alpha = 0$.
In both cases, we show that for every feasible instance, there is an optimal layout with certain structural properties.
For all convex instances, we obtain a natural ordering of the sources by going counterclockwise around the convex hull.
This yields a \emph{cyclic ordering} with no clear start and end.
We can derive a \emph{linear ordering} from the cyclic order by designating the start source.
In a layout of a convex instance with one intermediate layer, we call a Steiner vertex $v$ \emph{$k$-consecutive} with respect to a given ordering if the children of $v$ can be partitioned into at most $k$ sequences of consecutive sources. 
A layout $L$ is \emph{$k$-consecutive} if every Steiner vertex in $L$ is $k$-consecutive.
Figure~\ref{fig:3-consecutive} shows an example of a $3$-consecutive layout.

We show that every source-equally-spaced circular instance with one intermediate layer and one sink has an optimal layout that is $1$-consecutive with respect to the cyclic order by showing that such a layout exists for every linear ordering.
Building on this property, we propose a dynamic program for such instances.

However, this property does not directly translate to convex instances in general.
In fact, there is a convex instance with one intermediate layer, one sink, unlimited capacity and $\alpha = 0$, where no optimal layout is $2$-consecutive with respect to any cyclic ordering.
However, every such instance has an optimal layout that is $3$-consecutive with respect to the cyclic order.

\begin{figure}
  \centering
  \includegraphics[page=6]{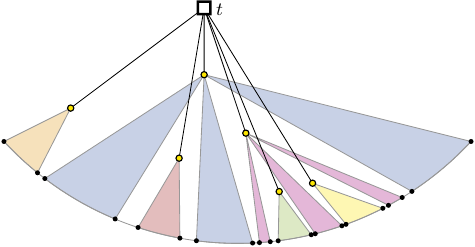}
  \caption{An example for a $3$-consecutive layout for an instance with one intermediate layer and one sink.}
  \label{fig:3-consecutive}
\end{figure}
Using this property, we formulate a dynamic program for convex instances with one intermediate layer, one sink, unlimited capacity and $\alpha = 0$.

\subsubsection{Solving Source-Equally-Spaced Circular Instances with One Intermediate Layer and One Sink}
\label{sec:convex:circle}
Compared to the previous hardness results for group-equally-spaced circular instances with one intermediate layer, we now make the two additional assumptions that there is a single sink $t$ and that the instance is additionally source-equally-spaced.
This means that the angle between adjacent sources along the circle is $\frac{2\pi}{n}$, where $n$ is the number of sources.

We first show that there is an optimal layout for every source-equally-spaced instance $I$ with one intermediate layer and one sink where every Steiner vertex is $1$-consecutive.
We denote the optimal cost of connecting $k$ consecutive sources to the sink via a single Steiner vertex in a layout for $I$ by $w_k$.
Note that this is well-defined in the sense that $w_k$ is the same for every interval of $k$ consecutive sources since the sources are evenly spaced.

%
\begin{lemma}
  \label{lem:convex:circular:contiguous}
  For every feasible source-equally-spaced circular instance $I = (S, \set{t}, C, \alpha)$ with one intermediate layer and one sink and every linear ordering of $S$, there is an optimal valid layout that is $1$-consecutive.
\end{lemma}
\begin{proof}
    \begin{figure}
    \centering
    \includegraphics[page=5]{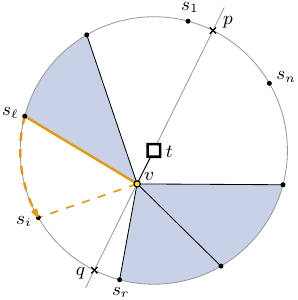}
    \caption{A Steiner vertex $v$ that is not $1$-consecutive. The cost decreases if $v$ is connected to $s_i$ instead of $s_\ell$.}
    \label{fig:circle-1-consecutive}
  \end{figure}
  We first show that the cost for a Steiner vertex $v$ with $k$ children is at least $w_k$, which is the optimal cost for a $1$-consecutive Steiner vertex with $k$ children.
  Let $L$ be a layout of $I$, and let $v$ be a Steiner vertex with $k$ children in $L$.
  We assume that the embedding cost for $v$ is optimal.
  If $v$ is $1$-consecutive, we are done.
  Otherwise $v$ is $k$-consecutive for some $k \geq 2$.
  If $v$ is located at the sink, we copy $v$ for each child such that each child has a distinct parent, and we are done.

  Assume in the following that $v$ is not located at the sink.
  Consider the line through $t$ and $v$ that cuts the circle into two halves.
  We denote the intersection with the circle that has larger distance to $v$ by $p$ and the other intersection by $q$.
  Enumerate the sources in $S$ in counterclockwise order starting at the first source after $p$ by $s_1, \ldots, s_n$.
  Since $L$ is not $1$-consecutive, there are three sources $s_\ell, s_i, s_r$ with $\ell < i < r$ such that $s_\ell$ and $s_r$ are children of $v$, but $s_i$ is not.
  Such a situation is depicted in Figure~\ref{fig:circle-1-consecutive}.
  If $s_i$ comes before $q$ along the circle, then we connect $v$ to $s_i$ instead of to $s_\ell$.
  Otherwise, we connect $v$ to $s_i$ instead of $s_r$.
  We claim that in both cases, the cost decreases.
  Assume without loss of generality that $s_i$ comes before $q$ along the circle.
  The distance between $v$ and a source $s \in S$ is $\dist(v, s) = \sqrt{r^2 + \dist(v,t)^2 - 2r\dist(v,t)\cos(\beta)}$, where $r$ is the radius of the circle and $\beta$ is the angle in $t$ between $v$ and $s$.
  As $\dist(v, s)$ decreases with smaller $\beta$ and the corresponding angle is smaller for $s_i$ than for $s_\ell$, it is $\dist(v, s_i) < \dist(v, s_\ell)$.
  Thus, the cost for $v$ in $L$ is not optimal, which is a contradiction.
  This means that a Steiner vertex with $k$ children is cost-optimal if it is $1$-consecutive, i.e., if the cost for the incident edges of $v$ in $L$ is $w_k$.

  These optimal building blocks allow us to construct a new layout $L'$.
  Assume that $L$ is an optimal valid layout for $I$ with $m$ Steiner vertices $v_1, \dots, v_{m}$, where $v_i$ has $k_i$ children for every $i \in [m]$.
  By the previous claim, the cost for $L$ is at least $\sum_{i = 1}^{m} w_{k_i}$.
  The layout $L'$ where $v_i$ is connected to $k_i$ consecutive sources costs exactly $\sum_{i = 1}^{m} w_{k_i}$ and is thus a $1$-consecutive valid layout with optimal cost.
\end{proof}
We can now formulate a dynamic program that solves source-equally-spaced circular instances with one intermediate layer and one sink.
Let $I$ be such an instance with sources $S = \set{s_1, \ldots, s_n}$, one intermediate layer and capacity $c_1$ in the intermediate layer.
Without loss of generality, we assume $c_1 \leq n$.
Note that if the capacity of the sink is less than the number of sources, then $I$ is trivially infeasible.

First, we compute $w_k$ for every $k \in [c_1]$, the cost of connecting $k$ consecutive sources to the same parent.
This gives us building blocks for a straightforward dynamic program that computes the optimal cost for connecting the sources $s_1, \ldots, s_i$ to the sink for increasing $i \in [n]$.
We denote this cost by $\DP[i]$ for every $i \in [n]$.
As our base case, we use the empty interval and clearly $\DP[0] = 0$.
Assume that we know $\DP[i]$ for all $i < j$.
An optimal layout on $s_1,\dots,s_j$ can be broken down into the subtree containing $s_j$ and an optimal layout of the first $j-k$ sources.
Let $k$ be the size of the subtree containing $s_j$.
Clearly, the cost of such a solution is $\DP[j-k] + w_k$.
Thus, the optimal cost to connect the sources $s_1, \ldots, s_j$ is the minimum over all subtree sizes up to the capacity $c_1$, i.e., $\DP[j] = \min_{k \leq c_1} \DP[j - k] + w_k$.
After running the dynamic program, $\DP[n]$ contains the optimal cost to connect all sources in $S$.
Using the values $\DP[i]$, it is easy to reconstruct an optimal valid layout that is $1$-consecutive.
The correctness of this dynamic program follows directly from Lemma~\ref{lem:convex:circular:contiguous}.

However, there is a caveat to this dynamic program, as it relies on determining the cost $w_k$ of connecting consecutive sources to one parent.
This boils down to finding the Weber point, which can only be done numerically.
Assuming that we get an $(1+\varepsilon)$-approximation for $w_k$ in $\Tweber\left(n, \varepsilon\right)$ time, this approximation factor directly translates to the result of our dynamic program, in the sense that it approximates the cost for an optimal layout within a factor of $1 + \varepsilon$.
\begin{lemma}
  \label{lemma:solve-circular-instance}
  Source-equally-spaced circular instances of \flamecast with one intermediate layer and one sink can be approximated within a factor of $1 + \varepsilon$ in $O(n \cdot \Tweber\left(n, \varepsilon\right) + n^2)$ time.
\end{lemma}
\begin{proof}
  We first bound the approximation factor.
  Let $\DP^*[i]$ be the optimal cost for connecting the sources $s_1, \ldots, s_i$,
  and let $w_i^*$ be the cost for connecting $i$ consecutive sources.
  Moreover, let $\DP[i]$ and $w_i$ be the corresponding values computed in our dynamic program.
  We show by induction over $i$ that $\DP[i] \leq (1 + \varepsilon) \DP^*[i]$.
  The values $w_i$ are computed using an approximation algorithm for the Weber problem for $i \in [c]$, i.e., it is $w_i \leq (1 + \varepsilon) \cdot w_i^*$.
  Assume now that every $\DP[i]$ is a $(1+ \varepsilon)$-approximation for $\DP[i]^*$ for every $i < j$.
  Then it is
  \begin{align*}
    \DP[j] &= \min_{k \leq c} \DP[j - k] + w_k = \min_{k \leq c} (1 + \varepsilon) \cdot \DP^*[j - k] + (1 + \varepsilon) \cdot w^*_k \\
  &= (1 + \varepsilon) \cdot \min_{k \leq c} \DP^*[j - k] + w^*_k = (1 + \varepsilon) \cdot \DP^*[j] .
  \end{align*}
  Thus, we have $\DP[n] \leq (1 + \varepsilon) \DP^*[n]$.

  For the running time, the cost for connecting $k$ consecutive sources to one parent takes $\Tweber(k, \varepsilon) \leq \Tweber(c_1, \varepsilon)$ time.
  The dynamic program itself goes over all $i \in [n]$ and considers every possible size of the subtree containing the last source in every step.
  Thus, the dynamic program has a running time of $O(n \cdot c_1)$.
  In the uncapacitated case, this results in a total of $O(n \cdot \Tweber\left(n, \varepsilon\right) + n^2)$ running time.
\end{proof}
Note that the previous lemma uses an algorithm that solves the Weber problem approximately as a black box.
If there is an Weber problem algorithm that determines the Weber point optimally, then our dynamic program is exact.
Cohen et al.~\cite{cohen2016} have shown that a $(1 + \varepsilon)$-approximate solution for the Weber problem in the plane with $n$ points can be computed in $O\left(n\log^3\left(\frac{1}{\varepsilon}\right)\right)$ time.
Thus, we get the following corollary.
\begin{corollary}
  \label{cor:circular-dp-running-time}
  Source-equally-spaced circular instances of \flamecast with one intermediate layer and one sink can be approximated within a factor of $1 + \varepsilon$ in $O\left(n^2\log^3\left(\frac{1}{\varepsilon}\right)\right)$ time, where $n$ is the number of sources.
\end{corollary}
As a consequence, we can compute a $(1 + \frac{1}{n^2})$-approximate valid layout in $O(n^2\log^3(n))$ time for source-equally-spaced circular instances, while on general circular instances, such an approximation ratio is impossible to reach if P $\neq$ NP by Theorem~\ref{thm:circular-approx-hard}.

\subsubsection{Structural Properties of Optimal Layouts for Convex Instances}
In this section, we consider structural properties of optimal layouts for convex instances with one intermediate layer, one sink, unlimited capacity and $\alpha = 0$.

We say that two subtrees \emph{interleave} if there is a sequence of sources $a_1,b_1,a_2,b_2$ in cyclic order such that $a_1$ and $a_2$ share a parent $v_a$ and $b_1$ and $b_2$ share a parent $v_b \neq v_a$.
Since by Lemma~\ref{lem:properties:no crossing}, edges incident to sources do not cross in an optimal layout, it follows directly that subtrees do not interleave in an optimal layout for convex instances.

In the following, we investigate the consecutivity of optimal layouts.



\subparagraph*{No Optimal 2-Consecutive Layout}
%
\begin{figure}
  \centering
  \subcaptionbox{
    Instance $I$
    \label{fig:convex:two breaks:base}
  }[.49\textwidth]{
    \includegraphics[page=1]{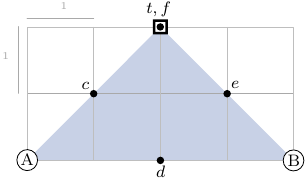}
  }
  \hfill
  \subcaptionbox{
    Optimal layout with cost $4+4\sqrt{2}\approx 9.66$
    \label{fig:convex:two breaks:opt}
  }[.49\textwidth]{
    \includegraphics[page=2]{figures/two_breaks.pdf}
  }
  \hfill
  \subcaptionbox{
    Alternative topology with cost $\approx 9.80$.
    \label{fig:convex:two breaks:first}
  }[.49\textwidth]{
    \includegraphics[page=3]{figures/two_breaks.pdf}
  }
%
  \caption{
    An instance for which no optimal layout is $2$-consecutive.
    For clarity, the edge between $A$ and $t$ is drawn arched, the actual connection is straight.
    The source $f$ coincides with sink $t$ which is marked by an empty square.
    \label{fig:convex:two breaks}
  }
\end{figure}

We demonstrate that there are convex instances with one intermediate layer, one sink, unlimited capacity and $\alpha = 0$ for which there is no $2$-consecutive optimal layout for any cyclic ordering.
In our example instance, the sources are laid out in a triangle.
We place a single source on each side of the triangle and multiple sources in each corner.
Choosing positions suitably, this setup ensures that the sources on the triangle sides connect to a common Steiner vertex in the triangle center and each corner has a separate Steiner vertex for its respective sources.
The construction is detailed in the proof of the following lemma and illustrated in \cref{fig:convex:two breaks}.

\begin{lemma}
  \label{lem:convex:three intervals}
  There is a convex instance $I = (S, \set{t}, (\infty, \infty, \infty), 0)$ of \flamecast with sources $S$ and a single sink $t$ for which any optimal layout contains a Steiner vertex that is not $2$-consecutive.
\end{lemma}
\begin{proof}
  We describe an instance $I$ that has the desired property.
  The sources $S$ and the sink $t$ are located on an isosceles triangle of height $2$ and base width $4$; see \cref{fig:convex:two breaks:base} for an illustration.
  The sources $c, d$ and $e$ are placed at the center of the sides such that $d$ is at the center of the base and $c$ and $e$ are placed at the center of the left and right leg respectively.
  A group of four sources is placed at each base vertex.
  We denote the groups by $A$ and $B$ where $A$ is on the left and $B$ is on the right vertex of the triangle.
  A single source $f$ and the sink $t$ are both placed at the apex.

  Consider the layout $L^*$ that has four Steiner vertices with children $\set{c, d, e}$, $A$, $B$ and $\set{f}$ respectively, which has cost $4\sqrt2 + 4 \leq 9.66$ (see \cref{fig:convex:two breaks:opt}).
  We show that any other topology costs more than $L^*$ by proving the following claims in the specified order.
  \begin{enumerate}
    \item\label{item:w-sink} In an optimal layout, $f$ is directly connected to the sink.
    \item\label{item:AB-distinct} In an optimal layout, sources in $A$ share a parent $v_A$ that is located at $A$, and sources in $B$ share a parent $v_B \neq v_A$ that is located at $B$.
    \item\label{item:ABxyz-distinct} If a source $s \in \set{c, d, e}$ has $v_A$ or $v_B$ as parent, then reconnecting $s$ directly to $t$ does not increase the cost.
    \item\label{item:xyz-not-sink} In an optimal layout, there is no source $s \in \set{c, d, e}$ that is directly connected to $t$.
  \end{enumerate}
  It follows directly from the last two claims that no vertex in $\set{c, d, e}$ shares a parent with $A$ or $B$ in an optimal layout, and none of them is directly connected to the sink.
  Thus, in an optimal layout, the vertices in $\set{c, d, e}$ share a parent.
  Together with the first two claims, this proves that the topology of any optimal layout equals the topology of $L^*$.

  Claim~\ref{item:w-sink} is obvious.
  For Claim~\ref{item:AB-distinct}, we first note that sources that share the same position also share a parent in an optimal layout by Lemma~\ref{cor:no-split}.
  If $A$ and $B$ share a parent $v_{AB}$, then the cost for connecting $A$ and $B$ to $v_{AB}$ is $|A|\dist(A, v_{AB}) + |B|\dist(B, v_{AB}) \geq 4 \dist(A, B) = 16 > \cost(L^*)$.
  Thus, it is $v_A \neq v_B$ in an optimal layout.
  It remains to show that $v_A$ is located at $A$.
  The Steiner vertex $v_A$ may be connected to at most seven sources $A$, $c$, $d$, and $e$ and one sink $t$.
  Since $|A| = 4$, it follows from Lemma~\ref{lem:properties:majority-weber-point} that $v_A$ is located at $A$.
  By symmetry, the same holds for $v_B$ and $B$.
  To prove Claim~\ref{item:ABxyz-distinct}, we assume that there is a source $s \in \set{c, d, e}$ that is connected to $v_A$ in an optimal layout $L$.
  By construction, it is $\dist(s, t) \leq \dist(s, v_A)$, where $\dist(s, v_A)$ equals $\dist(s, A)$ by the previous claim.
  Thus, the cost of $L$ does not increase if $s$ is reconnected to $t$.
  The same holds if $s$ is connected to $v_B$.

  It remains to show Claim~\ref{item:xyz-not-sink}, i.e., that in an optimal layout, no source in $\set{c, d, e}$ is directly connected to the sink.
  In $L^*$, $c$, $d$ and $e$ share a Steiner vertex $v$, and the cost for the incident edges of $v$ is $4$.
  If $c$, $d$ and $e$ are connected directly to the sink, then the cost for connecting them is $2 + 2\sqrt{2} > 4$, which is not optimal.
  Assume now that $e$ is directly connected to the sink, while $c$ and $d$ share a parent $v_{cd}$ (see \cref{fig:convex:two breaks:first}).
  By the geometry of the Weber point, $v_{cd}$ is located $\tan(\ang{30})=\frac{1}{\sqrt3}$ units from the line $dt$ towards $c$.
  Accordingly, we can compute the partial cost of this layout as $2 \cdot \sqrt{1 + \frac{1}{3}} + (1 - \frac{1}{\sqrt3}) + \sqrt2 = \sqrt2 + \sqrt3 + 1 > 4$.
  Thus, it is also not optimal to only connect $e$ directly to the sink.
  By symmetry, the same holds for $c$.
  Thus, in an optimal layout, $c$ and $e$ share a parent $v_{ce}$ that is not located at $t$.

  Assume now that $d$ is not connected to $v_{ce}$, but directly to $t$.
  By symmetry, $v_{ce}$ is located on the line $dt$.
  Thus, $\dist(d, v_{ce}) < \dist(d, t)$, and it is better to connect $d$ to $v_{ce}$ than to $t$.

  To conclude, $L^*$ is the unique optimal layout of the given instance.
  But then, the parent of $c$, $d$ and $e$ is not $2$-consecutive, regardless of the ordering.

\end{proof}

\subparagraph*{Optimal 3-Consecutive Layout}
While not every convex instance with one intermediate layer, one sink, unlimited capacity and $\alpha = 0$ has an optimal layout that is $2$-consecutive, we show every such instance has an optimal layout that is $3$-consecutive with respect to any cyclic ordering.
This claim follows from the following geometric property of convex instances that restricts the angle between Steiner vertices.

\begin{lemma}
  \label{lem:convex:angle}
  Let $P$ be a convex polygon with boundary $\delta P$ and let $u, v, w \in P$.
  If there are points $a, b_1, c, b_2 \in \delta P$ in this cyclic order such that $a \in \vor(u)$,  $b_1,b_2 \in \vor(v)$, and $c \in \vor(w)$, then the angle $\angle uvw \geq \ang{90}$.
\end{lemma}
\begin{proof}
      \begin{figure}
    \centering
    \includegraphics{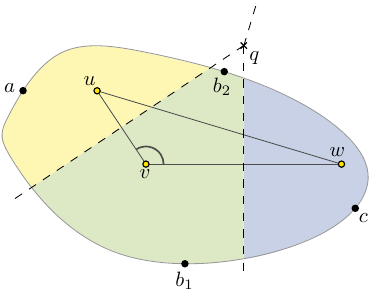}
    \caption{The setting of \cref{lem:convex:angle}. The intersection of the Voronoi cells of the vertices $u$, $v$ and $w$ with the polygon is colored.
    The angle $\angle uvw$ is at least $\ang{90}$.}
    \label{fig:convex-angle-steiner-vertex}
  \end{figure}
  The Voronoi cells of $u$, $v$ and $w$ are convex regions that intersect in the point $q$, the circumcenter of the triangle $\triangle uvw$.
  The setting is depicted in \cref{fig:convex-angle-steiner-vertex}.
  By the convexity of $\vor(v)$, the line segment $b_1 b_2$ is also in $\vor(v)$ and thus $\vor(v)$ separates $\vor(u) \cap P$ and $\vor(w) \cap P$.
  Accordingly, $q$ must be outside of $P$ and, in particular, $v$ and $q$ cannot be on the same side of the line $uw$.
  Thales theorem and its generalization, the inscribed angle theorem, state that this is the case if and only if the angle $\angle uvw \geq \ang{90}$.
\end{proof}
%
Using the previous lemma, we can now prove that every convex instance with one intermediate layer, one sink, unlimited capacity and $\alpha = 0$ has an optimal layout that is $3$-consecutive.
\begin{lemma}
  \label{lem:convex:max gaps}
  Let $I = (S, \set{t}, (\infty, \infty, \infty), 0)$ be a convex instance of \flamecast, and let $L$ be a layout of $I$.
  Then, there is a layout $L^*$ that is $3$-consecutive for any cyclic ordering with $\cost(L^*) \leq \cost(L)$.
\end{lemma}

\begin{proof}
  \begin{figure}
  \centering
  \subcaptionbox{
    A Steiner vertex $v_A$ that is not $3$-consecutive and the relevant sources and Steiner vertices.
    \label{fig:convex:three-consec-1}
  }[.49\textwidth]{
    \includegraphics[page=2]{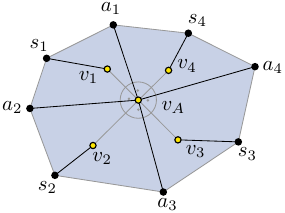}
  }
  \hfill
  \subcaptionbox{
    Possible positions of the children of $v_A$.
    The colored areas are the Voronoi cells of the Steiner vertices, respectively.
    For each source $a \in A$, there is another Steiner vertex that has the same distance to $a$ as $v_A$.
    \label{fig:convex:three-consec-2}
  }[.49\textwidth]{
    \includegraphics[page=3]{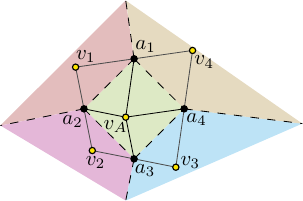}
  }
  \caption{
    The setting of \cref{lem:convex:max gaps}.
    \label{fig:convex:three-consec}
  }
\end{figure}

  If $L$ is $3$-consecutive, we are done.
  Assume $L$ is not $3$-consecutive, and that $L$ has optimal cost.
  Then, there is a Steiner vertex $v_A$ that is not $3$-consecutive.
  Let $a_1,s_1,a_2,s_2,a_3,s_3,a_4,s_4$ be a cyclically ordered subsequence of sources in $L$ such that $x_i$ is a child of $v_A$ and let $v_i \neq v_A$ be the parent of $s_i$, respectively.
  Let $A = \set{a_1, a_2, a_3, a_4}$.
  \cref{fig:convex:three-consec-1} depicts the setting.
  We show that if $v_A$ is not $3$-consecutive in an optimal layout, then the positions of the sources in $A$ can be uniquely determined.
  Moreover, they can be reconnected such that $v_A$ is $3$-consecutive without increasing the cost.

  We first determine the positions of $v_1, v_2, v_3, v_4$ relative to each other and to $v_A$.
  Since we assume $L$ to have minimum cost, there are no interleaving subtrees and thus, $v_i \neq v_j$ for $i \neq j$.
  For $\alpha=0$ and unlimited capacity, every source connects to the closest Steiner vertex in an optimal layout, i.e. it must be in the Voronoi cell of its parent.
  Consider the Steiner vertex $v_A$ and any Steiner vertices $v_i$ and $v_j$.
  Then, $s_i$, $a_{i + 1}$, $s_j$ and $a_{j+1}$ appear in this cyclic order on the convex hull of $S$ with $a_{i + 1}, a_{j + 1} \in V(v_A)$, and $s_i \in V(v_i)$ and $s_j \in V(v_j)$.
  By \cref{lem:convex:angle}, the angle $\angle v_i v_A v_j \geq \ang{90}$ for any $i \neq j$.
  Thus, the angles $\angle v_1v_2$, $\angle v_2v_3$, $\angle v_3v_4$ and $\angle v_4v_1$ are exactly $\ang{90}$, and the segments $v_1v_3$ and $v_2v_4$ meet at a right angle at $v_A$.

%
  We now determine possible positions for the sources in $A$.
  By \cref{lem:properties:convex steiner}, Steiner vertices are located in the convex hull of their neighborhood and since $I$ is convex, there are no sources in $S$ located strictly inside the quadrilateral defined by $v_1v_2v_3v_4$.
  Furthermore, the sources in $A$ lie in the Voronoi cell of $v_A$, whose boundary is determined by the perpendicular bisectors of the four line segments $v_Av_1$, $v_Av_2$, $v_Av_3$ and $v_Av_4$.
  By Thales's theorem, the vertices of the Voronoi cell $\vor(v_A)$ are the midpoints of each line segment.
  This can be seen in \cref{fig:convex:three-consec-2}.
  Thus, these midpoints are the only possible positions for the sources in $A$.
  As vertices of the Voronoi cell $\vor(v_A)$, they have the same distance to $v_A$ and two other Steiner vertices in $\set{v_1, v_2, v_3, v_4}$.

  We obtain a new layout $L'$ by connecting each source in $A$ to a closest vertex in $\set{v_1, v_2, v_3, v_4}$, which does not increase the cost by the previous argument.
  Since we assume unlimited capacities, $L'$ is still valid.
  Repeating this argument as long as there is a Steiner vertex that is not $3$-consecutive, yields a $3$-consecutive layout $L^*$ with the same cost as $L$.
\end{proof}
Given a $k$-consecutive layout, cutting the order between two adjacent sources in the cyclic order that have distinct parents clearly preserves $k$-consecutivity with respect to the resulting linear order.
This argument directly implies the following statement.
\begin{corollary}
  \label{cor:linear-order-k-consec}
    Let $I = (S, \set{t}, (\infty, \infty, \infty), 0)$ be a convex instance of \flamecast.
    There is a linear order on $S$ and an optimal layout $L$ for $I$ such that $L$ is $3$-consecutive with respect to the linear order.
\end{corollary}

\subsubsection{Solving Convex Instance with One Intermediate Layer, One Sink, Unlimited Capacity and \texorpdfstring{\boldmath$\alpha = 0$}{𝛼 = 0}}


We present a dynamic program that computes a $(1 + \varepsilon)$-approximation for a convex instance $I = (S, \set{t}, (\infty, \infty, \infty), 0)$.
To find a 3-consecutive optimal layout for the cyclic ordering, we introduce the notion of source intervals.
By computing a 3-consecutive optimal layout for every start source, one of those layouts must be the optimal layout for the cyclic ordering.
We call an ordered set of consecutive sources $(s_i, s_{i+1}, \dots, s_j)$ a \emph{source interval}.
We consider maximal source intervals where $s_i = s_{j+1}$ to be distinct if they do not have the same start and end.
The maximal source intervals thus correspond to the linear orderings of the sources.
Note that a source interval can be empty.

Let $L$ be a 3-consecutive optimal layout with respect to the cyclic ordering.
Then there is a linear ordering, and accordingly a maximal source interval, for which $L$ is 3-consecutive.
We describe a dynamic program that computes a 3-consecutive optimal layout for every source interval.

Before we formulate the dynamic program, we introduce some helpful notation.
For a subset of sources $A \subseteq S$, we denote the cost for connecting $A$ to the sink using a single Steiner vertex by $\cost(A)$.
If $A$ is empty, then the cost is $0$.
Moreover, we say a tuple $(A_1, A_2, \ldots, A_k)$ is an \emph{interval partition} of the source interval $A$ if each part is a source interval, the disjoint union of all parts is equal to $A$, the parts are sorted with respect to the linear ordering, and $A_i \neq A$ for $i \in [k]$.
Note that this means that parts may be empty.

The core idea of the dynamic program is to find optimal $3$-consecutive layouts for smaller source intervals and to combine these to build optimal $3$-consecutive layouts for larger source intervals.
Consider a $3$-consecutive layout $L$ with the lowest cost for connecting the sources in the source interval $(s_1, \ldots, s_n)$.
We know that in this case there must be an interval partition $(A_1, S_1, A_2, S_2, A_3)$ of the sources in $L$ such that all vertices in $A_1, A_2$ and $A_3$ have the same parent which and that this parent is distinct from the parents of all vertices in $S_1$ and $S_2$.
Note that $A_1, A_2$ and $A_3$ may all be empty; in this case the layout is composed of optimal layouts for two smaller source intervals.

In particular, the structure of $L$ is determined by whether $s_1$ and $s_n$ share a parent $v$ or whether they have distinct parents.
In case they share a parent, we distinguish by the number of source intervals needed to cover the children of $v$.
This leads to four structural cases are illustrated in \cref{fig:3-consecutive-dp}.

\begin{figure}
  \centering
  \includegraphics[page=7]{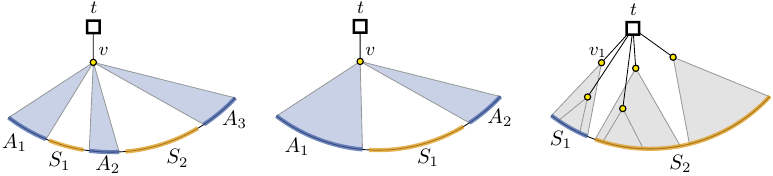}
  \caption{Three of the four possible cases in a $3$-consecutive layout; the case where all sources in the interval are connected to a single Steiner vertex is missing.
  Depending on if the first and the last source of the interval share a parent $v$, we partition the source interval differently.}
  \label{fig:3-consecutive-dp}
\end{figure}
%



We can now formulate our dynamic program, which determines for each source interval $A \subseteq S$ the lowest cost $\DP[A]$ of a $3$-consecutive layout that connects the sources in $A$.
If $A$ is empty, then $\DP[A] = 0$, and if $A$ only consists of one source $a$, then $\DP[A] = \dist(a, t)$.
Otherwise, we can try all possible interval partitions $(A_1, S_1, A_2, S_2, A_3)$ of $A$.
%
This gives us the following dynamic program.
\begin{align*}
\DP[A] &= \min\left\{
        \min_{\substack{(A_1, S_1, A_2, S_2, A_3)\\ \text{partition of $A$}}} \DP[S_1] + \DP[S_2] + \cost(A_1 \cup A_2 \cup A_3) 
    \right\}
\end{align*}
The cost for an optimal 3-consecutive layout with respect to the cyclic ordering is the minimum entry of $\DP[S_i]$ over all maximal source intervals $S_i$.
As in the circular case, the dynamic program relies on determining the cost for connecting a set of sources using a single Steiner vertex.
Assuming that we get a $(1 + \varepsilon)$-approximation for $\cost(A)$ with $A \subseteq S$, our dynamic program also is a $(1 + \varepsilon)$-approximation.
Let $\Tweber(|A|,\varepsilon)$ be the time to compute a $(1 + \varepsilon)$-approximation for $\cost(A)$.

\begin{theorem}
  \label{theorem:dp}
  Convex instances of \flamecast with one intermediate layer, one sink, unlimited capacity and $\alpha = 0$ can be approximated within a factor of $1 + \varepsilon$ in $O(n^6 \cdot T_{Weber}(n, \varepsilon))$ time.
\end{theorem}
\begin{proof}
  We first bound the approximation factor.
  Let $\DP^*[A]$ be the optimal cost for connecting the source interval $A$,
  and let $\cost^*(S')$ be the cost for connecting the sources in set $S'$ using a single Steiner vertex.
  Moreover, let $\DP[A]$ and $\cost(S')$ be the corresponding values computed in our dynamic program.
  We show by induction over the size of $A$ that $\DP[A] \leq (1 + \varepsilon) \DP^*[A]$.
  If $A$ is empty or contains one element, then it is $\DP[A] = \DP^*[A]$.
  The values $\cost(S')$ are computed using an approximation algorithm for the Weber problem, i.e., it is $\cost(S') \leq (1 + \varepsilon) \cdot \cost^*(S')$.
  Assume now that every $\DP[B]$ is a $(1 + \varepsilon)$-approximation for $\DP^*[B]$ for every $|B| < |A|$.
  Then it is
  \begin{align*}
      \DP[A] &= \min\Biggl\{
    \begin{aligned}
        &\min_{\substack{(S_1, S_2)\\ \text{partition of $A$}}} \DP[S_1] + \DP[S_2], \\
        &\min_{\substack{(A_1, S_1, A_2, S_2, A_3)\\ \text{partition of $A$}}} \DP[S_1] + \DP[S_2] + \cost(A_1 \cup A_2 \cup A_3)
    \end{aligned}
    \Biggr\}\\
   &\leq \min\Biggl\{
    \begin{aligned}
        &\min_{\substack{(S_1, S_2)\\ \text{partition of $A$}}} (1+ \varepsilon) (\DP^*[S_1] + \DP^*[S_2]), \\
        &\min_{\substack{(A_1, S_1, A_2, S_2, A_3)\\ \text{partition of $A$}}} (1+ \varepsilon) (\DP^*[S_1] + \DP^*[S_2] + \cost^*(A_1 \cup A_2 \cup A_3))
    \end{aligned}
    \Biggr\} \\
   &= (1+ \varepsilon) \DP^*[A].
  \end{align*}
  Thus, it is $\DP[S] \leq (1 + \varepsilon) \cdot \DP^*[S]$.

  The dynamic program goes over all $O(n^2)$ source intervals for $S$.
  For each source interval, it tries all possible interval partitions into five parts (taking $O(n^4)$ time).
  Computing $\cost(A_1 \cup A_2 \cup A_3)$ requires computing the Weber point in $\Tweber(n, \varepsilon)$ time.
  Thus, the dynamic program runs in $O(n^2 \cdot (n + n^4 \cdot \Tweber(n, \varepsilon))) = O(n^6 \Tweber(n, \varepsilon))$ time.
\end{proof}
With an algorithm by Cohen~\cite{cohen2016}, where $\Tweber(n, \varepsilon) \in O\left(n\log^3\left(\frac{1}{\varepsilon}\right)\right)$, we achieve the following running time.
\begin{corollary}
    Convex instances of \flamecast with one intermediate layer, one sink, unlimited capacity and $\alpha = 0$ can be $(1 + \varepsilon)$-approximated in $O\left(n^7\log^3\left(\frac{1}{\varepsilon}\right)\right)$ time.
\end{corollary}

\section{Conclusion}
We introduced the \flamecast problem, which captures the design of hierarchical networks under load-dependent edge costs.
We showed that \flamecast is NP-hard to approximate within a certain factor for various settings with different parameters and complemented these results with polynomial approximation algorithms for more restricted circular and convex instances with one intermediate layer and one sink.
In both cases, the key insight used by the dynamic programs is that optimal layouts for these restricted settings have specific structures.
In fact, our dynamic program can be easily generalized to all convex instances with one intermediate layer and one sink that admit a $k$-consecutive optimal layout, yielding a polynomial algorithm if $k$ is constant.
From this observation, the natural question arises which convex instances have $k$-consecutive layouts.
While we have shown that such instances with $\alpha = 0$ and unlimited capacity admit $3$-consecutive optimal layouts, it is still open if this also holds for any $\alpha > 0$.
If the capacity in the intermediate layer is limited, then there are convex instances with no $k$-consecutive optimal layouts for any constant $k$.
Here, it makes sense to impose some further (realistic) restrictions such as a minimum distance between any pair of sources, and to investigate consecutivity of optimal layouts under these restrictions.

Experiments suggest that source-equally-spaced circular instances with one sink may admit $1$-consecutive optimal layouts, even with more than one intermediate layer.
However, for multiple layers, the topology of a subtree is not fixed even if the sources are fixed.
Thus, to obtain the cost for connecting a specific set of sources in a naive approach, all possible topologies of a subtree have to be enumerated, which is clearly not polynomial.
Further research could investigate the complexity of convex instances with more than one intermediate layer.

\bibliography{flamecast.bib}

\begin{thebibliography}{10}

\bibitem{carlo2012}
Héctor~J. Carlo, Francisco Aldarondo, Priscilla~M. Saavedra, and Silmarie~N.
  Torres.
\newblock Capacitated continuous facility location problem with unknown number
  of facilities.
\newblock {\em Engineering Management Journal}, 24(3):24--31, 2012.
\newblock \href
  {https://arxiv.org/abs/https://doi.org/10.1080/10429247.2012.11431944}
  {\path{arXiv:https://doi.org/10.1080/10429247.2012.11431944}}, \href
  {https://doi.org/10.1080/10429247.2012.11431944}
  {\path{doi:10.1080/10429247.2012.11431944}}.

\bibitem{cohen2016}
Michael~B. Cohen, Yin~Tat Lee, Gary~L. Miller, Jakub~W. Pachocki, and Aaron
  Sidford.
\newblock Geometric median in nearly linear time.
\newblock {\em Proceedings of the forty-eighth annual ACM symposium on Theory
  of Computing}, 2016.
\newblock URL: \url{https://api.semanticscholar.org/CorpusID:7287824}.

\bibitem{sahin2007}
Güvenç \c{S}ahin and Haldun Süral.
\newblock A review of hierarchical facility location models.
\newblock {\em Computers \& Operations Research}, 34(8):2310--2331, 2007.
\newblock URL:
  \url{https://www.sciencedirect.com/science/article/pii/S0305054805002959},
  \href {https://doi.org/10.1016/j.cor.2005.09.005}
  {\path{doi:10.1016/j.cor.2005.09.005}}.

\bibitem{deberg2010}
Mark de~Berg and Amirali Khosravi.
\newblock Optimal binary space partitions in the plane.
\newblock In My~T. Thai and Sartaj Sahni, editors, {\em Computing and
  Combinatorics}, pages 216--225, Berlin, Heidelberg, 2010. Springer Berlin
  Heidelberg.

\bibitem{drezner2004}
Zvi Drezner and Horst~W Hamacher.
\newblock {\em Facility location: applications and theory}.
\newblock Springer Science \& Business Media, 2004.

\bibitem{garey2002}
Michael~R Garey and David~S Johnson.
\newblock {\em Computers and intractability}, volume~29.
\newblock wh freeman New York, 2002.

\bibitem{gokbayrak2017}
Kagan Gokbayrak and Ayse~Selin Kocaman.
\newblock A distance-limited continuous location-allocation problem for spatial
  planning of decentralized systems.
\newblock {\em Computers \& Operations Research}, 88:15--29, 2017.
\newblock URL:
  \url{https://www.sciencedirect.com/science/article/pii/S030505481730151X},
  \href {https://doi.org/10.1016/j.cor.2017.06.013}
  {\path{doi:10.1016/j.cor.2017.06.013}}.

\bibitem{gritzbach2022}
Sascha Gritzbach, Dominik Stampa, and Matthias Wolf.
\newblock Solar farm cable layout optimization as a graph problem.
\newblock {\em Energy Informatics}, 5(S1):Art.Nr. 25, 2022.
\newblock 37.12.02; LK 01.
\newblock \href {https://doi.org/10.1186/s42162-022-00200-z}
  {\path{doi:10.1186/s42162-022-00200-z}}.

\bibitem{hakimi1971}
S.~L. Hakimi.
\newblock Steiner's problem in graphs and its implications.
\newblock {\em Networks}, 1(2):113--133, 1971.
\newblock URL:
  \url{https://onlinelibrary.wiley.com/doi/abs/10.1002/net.3230010203}, \href
  {https://arxiv.org/abs/https://onlinelibrary.wiley.com/doi/pdf/10.1002/net.3230010203}
  {\path{arXiv:https://onlinelibrary.wiley.com/doi/pdf/10.1002/net.3230010203}},
  \href {https://doi.org/10.1002/net.3230010203}
  {\path{doi:10.1002/net.3230010203}}.

\bibitem{Karp1972}
Richard~M. Karp.
\newblock {\em Reducibility among Combinatorial Problems}, pages 85--103.
\newblock Springer US, Boston, MA, 1972.
\newblock \href {https://doi.org/10.1007/978-1-4684-2001-2_9}
  {\path{doi:10.1007/978-1-4684-2001-2_9}}.

\bibitem{kaufman1977}
Leon Kaufman, Marc~Vanden Eede, and Pierre Hansen.
\newblock A plant and warehouse location problem.
\newblock {\em Journal of the Operational Research Society}, 28(3):547--554,
  1977.

\bibitem{knuth1992}
Donald~E. Knuth and Arvind Raghunathan.
\newblock The problem of compatible representatives.
\newblock {\em SIAM Journal on Discrete Mathematics}, 5(3):422--427, 1992.
\newblock \href {https://arxiv.org/abs/https://doi.org/10.1137/0405033}
  {\path{arXiv:https://doi.org/10.1137/0405033}}, \href
  {https://doi.org/10.1137/0405033} {\path{doi:10.1137/0405033}}.

\bibitem{korupolu2000}
Madhukar~R. Korupolu, C.Greg Plaxton, and Rajmohan Rajaraman.
\newblock Analysis of a local search heuristic for facility location problems.
\newblock {\em Journal of Algorithms}, 37(1):146--188, 2000.
\newblock URL:
  \url{https://www.sciencedirect.com/science/article/pii/S0196677400911003},
  \href {https://doi.org/10.1006/jagm.2000.1100}
  {\path{doi:10.1006/jagm.2000.1100}}.

\bibitem{kuhn1955}
Harold~W. Kuhn.
\newblock The hungarian method for the assignment problem.
\newblock {\em Naval Research Logistics Quarterly}, 2(1-2):83--97, 1955.
\newblock URL:
  \url{https://onlinelibrary.wiley.com/doi/abs/10.1002/nav.3800020109}, \href
  {https://arxiv.org/abs/https://onlinelibrary.wiley.com/doi/pdf/10.1002/nav.3800020109}
  {\path{arXiv:https://onlinelibrary.wiley.com/doi/pdf/10.1002/nav.3800020109}},
  \href {https://doi.org/10.1002/nav.3800020109}
  {\path{doi:10.1002/nav.3800020109}}.

\bibitem{kuhn1962}
Harold~W. Kuhn and Robert~E. Kuenne.
\newblock An efficient algorithm for the numerical solution of the generalized
  weber problem in spatial economics.
\newblock {\em Journal of Regional Science}, 4(2):21--33, 1962.
\newblock \href {https://doi.org/10.1111/j.1467-9787.1962.tb00902.x}
  {\path{doi:10.1111/j.1467-9787.1962.tb00902.x}}.

\bibitem{ljubic2021}
Ivana Ljubić.
\newblock Solving steiner trees: Recent advances, challenges, and perspectives.
\newblock {\em Networks}, 77(2):177--204, 2021.
\newblock URL: \url{https://onlinelibrary.wiley.com/doi/abs/10.1002/net.22005},
  \href
  {https://arxiv.org/abs/https://onlinelibrary.wiley.com/doi/pdf/10.1002/net.22005}
  {\path{arXiv:https://onlinelibrary.wiley.com/doi/pdf/10.1002/net.22005}},
  \href {https://doi.org/10.1002/net.22005} {\path{doi:10.1002/net.22005}}.

\bibitem{love1988}
Robert~F Love, James~G Morris, and George~O Wesolowsky.
\newblock Facilities location.
\newblock 1988.

\bibitem{megiddo1984}
Nimrod Megiddo and Kenneth~J. Supowit.
\newblock On the complexity of some common geometric location problems.
\newblock {\em SIAM Journal on Computing}, 13(1):182--196, 1984.
\newblock \href {https://arxiv.org/abs/https://doi.org/10.1137/0213014}
  {\path{arXiv:https://doi.org/10.1137/0213014}}, \href
  {https://doi.org/10.1137/0213014} {\path{doi:10.1137/0213014}}.

\bibitem{munkres1957}
James Munkres.
\newblock Algorithms for the assignment and transportation problems.
\newblock {\em Journal of the Society for Industrial and Applied Mathematics},
  5(1):32--38, 1957.
\newblock \href {https://arxiv.org/abs/https://doi.org/10.1137/0105003}
  {\path{arXiv:https://doi.org/10.1137/0105003}}, \href
  {https://doi.org/10.1137/0105003} {\path{doi:10.1137/0105003}}.

\bibitem{narula1984}
Subhash~C. Narula.
\newblock Hierarchical location-allocation problems: A classification scheme.
\newblock {\em European Journal of Operational Research}, 15(1):93--99, 1984.
\newblock URL:
  \url{https://www.sciencedirect.com/science/article/pii/0377221784900523},
  \href {https://doi.org/10.1016/0377-2217(84)90052-3}
  {\path{doi:10.1016/0377-2217(84)90052-3}}.

\bibitem{nesetril2001}
Jaroslav Nešetřil, Eva Milková, and Helena Nešetřilová.
\newblock Otakar borůvka on minimum spanning tree problem translation of both
  the 1926 papers, comments, history.
\newblock {\em Discrete Mathematics}, 233(1):3--36, 2001.
\newblock Czech and Slovak 2.
\newblock URL:
  \url{https://www.sciencedirect.com/science/article/pii/S0012365X00002247},
  \href {https://doi.org/10.1016/S0012-365X(00)00224-7}
  {\path{doi:10.1016/S0012-365X(00)00224-7}}.

\bibitem{schultz1970}
George~P Schultz.
\newblock The logic of health care facility planning.
\newblock {\em Socio-Economic Planning Sciences}, 4(3):383--393, 1970.

\bibitem{Soland1974}
Richard~M. Soland.
\newblock Optimal facility location with concave costs.
\newblock {\em Operations Research}, 22(2):373--382, 1974.
\newblock URL: \url{http://www.jstor.org/stable/169594}.

\bibitem{tomizawa1971}
N.~Tomizawa.
\newblock On some techniques useful for solution of transportation network
  problems.
\newblock {\em Networks}, 1(2):173--194, 1971.
\newblock URL:
  \url{https://onlinelibrary.wiley.com/doi/abs/10.1002/net.3230010206}, \href
  {https://arxiv.org/abs/https://onlinelibrary.wiley.com/doi/pdf/10.1002/net.3230010206}
  {\path{arXiv:https://onlinelibrary.wiley.com/doi/pdf/10.1002/net.3230010206}},
  \href {https://doi.org/10.1002/net.3230010206}
  {\path{doi:10.1002/net.3230010206}}.

\bibitem{rahman2000}
Shams ur~Rahman and David~K. Smith.
\newblock Use of location-allocation models in health service development
  planning in developing nations.
\newblock {\em European Journal of Operational Research}, 123(3):437--452,
  2000.
\newblock URL:
  \url{https://www.sciencedirect.com/science/article/pii/S0377221799002891},
  \href {https://doi.org/10.1016/S0377-2217(99)00289-1}
  {\path{doi:10.1016/S0377-2217(99)00289-1}}.

\bibitem{volz2013}
M.~G. Volz, M.~Brazil, C.~J. Ras, K.~J. Swanepoel, and D.~A. Thomas.
\newblock The gilbert arborescence problem.
\newblock {\em Networks}, 61(3):238--247, 2013.
\newblock URL: \url{https://onlinelibrary.wiley.com/doi/abs/10.1002/net.21475},
  \href
  {https://arxiv.org/abs/https://onlinelibrary.wiley.com/doi/pdf/10.1002/net.21475}
  {\path{arXiv:https://onlinelibrary.wiley.com/doi/pdf/10.1002/net.21475}},
  \href {https://doi.org/10.1002/net.21475} {\path{doi:10.1002/net.21475}}.

\bibitem{voss1999}
S.~Voß.
\newblock {The Steiner tree problem with hop constraints}.
\newblock {\em Annals of Operations Research}, 86(0):321--345, January 1999.
\newblock URL:
  \url{https://ideas.repec.org/a/spr/annopr/v86y1999i0p321-34510.1023-a1018967121276.html},
  \href {https://doi.org/10.1023/A:1018967121276}
  {\path{doi:10.1023/A:1018967121276}}.

\bibitem{Westscott2023}
Alexander Westcott, Marcus Brazil, and Charl Ras.
\newblock Structural properties of minimum multi-source multi-sink steiner
  networks in the euclidean plane.
\newblock {\em J. Optim. Theory Appl.}, 197(3):1104–1139, April 2023.
\newblock \href {https://doi.org/10.1007/s10957-023-02203-6}
  {\path{doi:10.1007/s10957-023-02203-6}}.

\bibitem{zosin2002}
Leonid Zosin and Samir Khuller.
\newblock On directed steiner trees.
\newblock In {\em Proceedings of the Thirteenth Annual ACM-SIAM Symposium on
  Discrete Algorithms}, SODA '02, page 59–63, USA, 2002. Society for
  Industrial and Applied Mathematics.

\end{thebibliography}

\end{document}